\documentclass[11pt]{article}

\usepackage{mathtools}
\usepackage{amsmath, amsthm, amsfonts, amssymb}
\usepackage{enumitem}
\usepackage{graphicx}
\usepackage{colortbl}
\usepackage{tikz}
\usepackage[utf8]{inputenc}
\usepackage{esint}
\usepackage{mathrsfs}
\usepackage[T1]{fontenc}
\usepackage{mathrsfs}  
\usepackage{varwidth}
\usepackage{bbm} 
\usepackage{cite}
\usepackage{mathtools}
\usepackage{array}
\usepackage{dsfont}
\usepackage{ushort}
\usepackage{accents}
\usepackage{comment}
\usepackage{caption}
\usepackage{subcaption}
\usepackage[version=4]{mhchem}
\usepackage{lscape}

\newcommand\blfootnote[1]{%
	\begingroup
	\renewcommand\thefootnote{}\footnote{#1}%
	\addtocounter{footnote}{-1}%
	\endgroup
}

\newcommand{\lp}{\left(}
\newcommand{\rp}{\right)}

\def\dx{\,\mathrm{d}x}
\def\p{\,\partial}
\let\mc=\mathcal
\def\R{\mathbb{R}}

\def\1{\mathds{1}}

\let\eps\varepsilon

\usepackage[colorlinks=true,linkcolor=blue,citecolor=blue,urlcolor=blue,breaklinks]{hyperref}

\numberwithin{equation}{section}
\usepackage[titletoc, toc, page]{appendix} 

\usepackage{hyperref}
\usepackage{url}
\newtheorem{thm}{Theorem}[section]

\theoremstyle{definition}
\newtheorem{remark}[thm]{Remark}

\hypersetup{pdftitle={EcologicalInteractions}}
\hypersetup{pdfauthor={Viktoria Freingruber, Rebeca Gonzalez Cabaleiro  \& Havva Yoldaş }}

\author{Viktoria Freingruber$^{\ast \,\dagger}$
	\and Rebeca Gonzalez-Cabaleiro\footnote{Department of Biotechnology, Faculty of Applied Sciences, Delft University of Technology, Van der Maasweg 9, 2629 HZ Delft, The Netherlands. V.E.Freingruber@tudelft.nl \&  R.GonzalezCabaleiro@tudelft.nl}
	\and Havva Yolda\c{s}\footnote{Delft Institute of Applied Mathematics, Faculty of Electrical Engineering, Mathematics and Computer Science, Delft University of Technology, Mekelweg 4, 2628CD Delft, The Netherlands.  H.Yoldas@tudelft.nl}
}

\title{Ecological interactions and spatial dynamics in microbial aggregates: A novel modelling framework}

\makeindex

\begin{document}
	
	\maketitle
	
	\vspace{-10pt}

	\begin{abstract}
		
		\noindent
		We present a mathematical model based on a system of partial differential equations (PDEs) with cross-diffusion and reaction terms to describe ecological interactions between multiple bacterial species and substrates within microaggregates, where bacteria proliferate in response to substrate availability and undergo passive dispersal driven by population pressure gradients. The ecological interactions include interspecific competition for shared substrates, and commensalism, whereby one species benefits from the metabolic by-products of another.
        The main motivation comes from individual-based models (IBMs) of microbial aggregates, where simulations reveal that substrate-limited conditions can give rise to rich spatial patterns.
        Our numerical experiments demonstrate that our PDE-based model captures the key qualitative features of three verification scenarios that have previously been investigated with IBMs. Moreover, we formally derive a competition system from an on-lattice biased random walk, and establish local well-posedness for a parameter-symmetric subcase of it. We then formally analyse the travelling wave behaviour of this case in one spatial dimension and compare the minimal travelling wave speed with the wave speed measured in the simulations.
		
		\blfootnote{\emph{Keywords and phrases.} Multi-species microbial communities, Mathematical model, Cross-diffusion system, Competition, Commensalism, Travelling wave} 
		\blfootnote{\emph{2020 Mathematics Subject Classification.} 35Q92, 92D25, 92B99, 35M33, 92C75}
		
	\end{abstract}
	
	\section{Introduction} \label{sec:intro}
	
	Microorganisms represent an incredibly diverse and ubiquitous group of life forms, essential to the functioning of Earth's biochemical cycles \cite{FFD08}. They thrive in a vast range of environments—from soil and water to animal hosts and industrial bioreactors—where they coexist with numerous other microbial species \cite{LL20}. Within these complex ecosystems, bacterial species can engage in various ecological interactions: they may compete for limited resources, co-operate through metabolic exchange, harm one another through predation, or simply coexist without direct interaction \cite{AM89, B94, FR12, PF22}. These interspecies relationships play a crucial role in shaping the structure, biochemical function, and resilience of microbial communities. They can give rise to rich non-linear dynamics, including the emergence of spatial patterns and travelling wave–like expansions of biomass \cite{MASSG22, PT2008, VWKK2014, ZWWT2017, LDHD2014}.

	Microbial communities can spatially organise in a multitude of ways. In natural ecosystems, these range from loosely associated planktonic populations to spatially complex structured aggregates \cite{C20, CGC78, FWSSRK16}. In laboratory or industrial settings, bacterial cultures are often studied in the form of biofilms on surfaces or in suspended cultures such as bioreactors \cite{GGMVL18, V20, SJKH14}. In this article, we focus on microbial aggregates that form in controlled laboratory settings where the substrate supply is continuously regulated. This setting can be a bioreactor, where, even though these are usually well-stirred, bacterial populations can form granules or grow as biofilms on surfaces. These aggregates vary in size and provide localised environments where diffusion limitations, cell-to-cell interactions and spatial positioning strongly influence the local microbial composition.
	
	Mathematical models are a fundamental tool to investigate dynamics of microbial systems. Well-mixed bioreactor communities are traditionally modelled using systems of ordinary differential equations (ODEs), which describe average concentrations of substrates and population densities over time \cite{BW85, DH83, W90}. 
	However, these models fail to capture spatial effects, including diffusion-limited interactions and the persistence of microbial diversity through local niche advantages within microaggregates. In contrast, biofilms are typically studied using continuum models in the form of partial differential equations (PDEs) that describe the time evolution of the space-dependent population density \cite{EPL01, WEMNPRL05, WW98, WW99}. These models usually consider nutrient gradients and mechanical interactions between bacteria \cite{HL14, MFDPPE18, WZ10}. Another way to account for spatial effects in microbial systems is through individual-based models (IBMs) \cite{HCCPK16}. The IBMs have been used extensively for microbial aggregates \cite{MASSG22, MSSG23} and biofilms \cite{K04, KBW98, KPWL01, PKL04}.
	
	To illustrate the advantage of spatial models, consider, for instance, two competing populations: a growth strategist, which relies on a higher effective growth rate, and a yield strategist, which prioritises more efficient use of available substrate. A classical chemostat ODE model would predict the dominance of the growth strategist over the yield strategist \cite{BW85}. However, IBMs have demonstrated that under certain conditions the yield strategist can outcompete the growth strategist \cite{K04}, highlighting the crucial role of spatial structure in microbial systems. A major limitation of IBMs is that they are computationally demanding compared to ODE and PDE systems, as they typically track the position and size of each individual at every time step. Moreover, they are difficult to analyse mathematically since standard tools for studying long-term population dynamics cannot be applied directly. Hence, quantitative predictions of the asymptotic dynamics are rather limited.
	
	In this article, we introduce a general PDE-based model for microbial aggregates subject to various ecological interactions mediated by substrate availability. The resulting PDE system describes the spatio-temporal dynamics of an arbitrary number of interacting bacterial species and substrates and is therefore highly coupled, incorporating both cross-diffusion and reaction terms. Our motivation comes from the IBM presented in \cite{MASSG22}, describing the evolution of bacterial microaggregates through growth, division, and movement rules at the single-cell level coupled with a diffusion-reaction equation for the evolution of the substrate concentration within a controlled environment. Our goal is to address two main limitations of the IBM, namely, achieving analytical tractability and computational efficiency while retaining essential spatial information. Although an earlier work has explored connections between IBM and PDE formulations of biofilms \cite{KHE09}, most PDE models developed for biofilms have traditionally been introduced from a phenomenological perspective. We introduce a simplified, lattice-based version of the IBM in \cite{MASSG22} that retains the essential individual-level behaviours. From this discrete description, we then derive a general PDE system. To the best of our knowledge, the general PDE model is novel in this context. Although it is flexible enough to accommodate a wide range of ecological interactions, we initially focus on competition and commensalism (cf. Figure~\ref{fig:schematic}), which allows for a direct comparison with previous IBM results \cite{K04, MASSG22}. We perform numerical simulations for our PDE models for these cases and find that the qualitative dynamics agree very well with the results in \cite{K04, MASSG22}. In particular,
	\begin{figure}[ht!]
		\centering
		\begin{subfigure}[b]{0.49\textwidth}
			\centering
			\includegraphics[width=\textwidth]{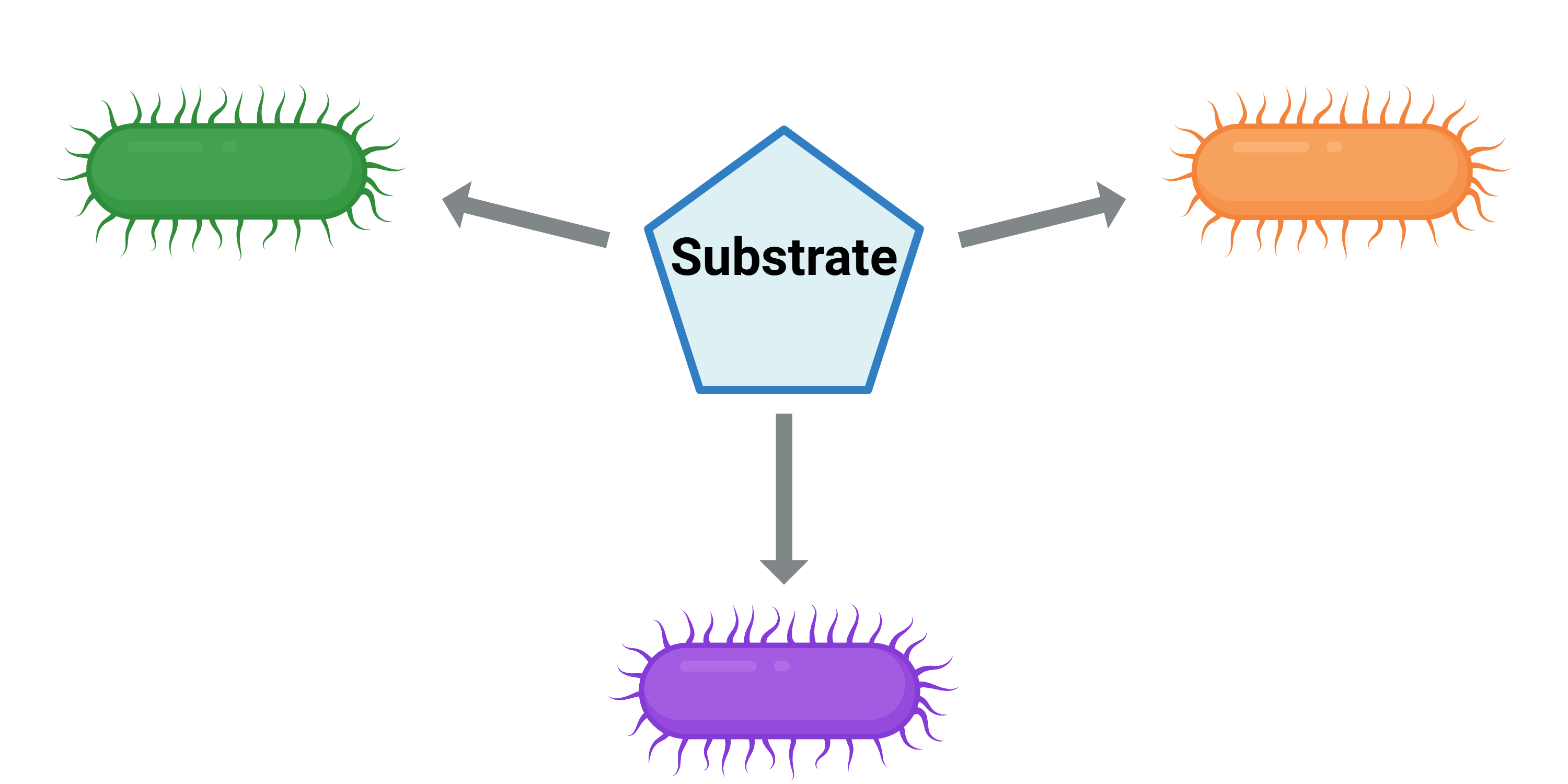}
			\caption{Competition}
			\label{fig:schematic_competition}
		\end{subfigure}	
		\hfill
		\begin{subfigure}[b]{0.49\textwidth}
			\centering
			\includegraphics[width=\textwidth]{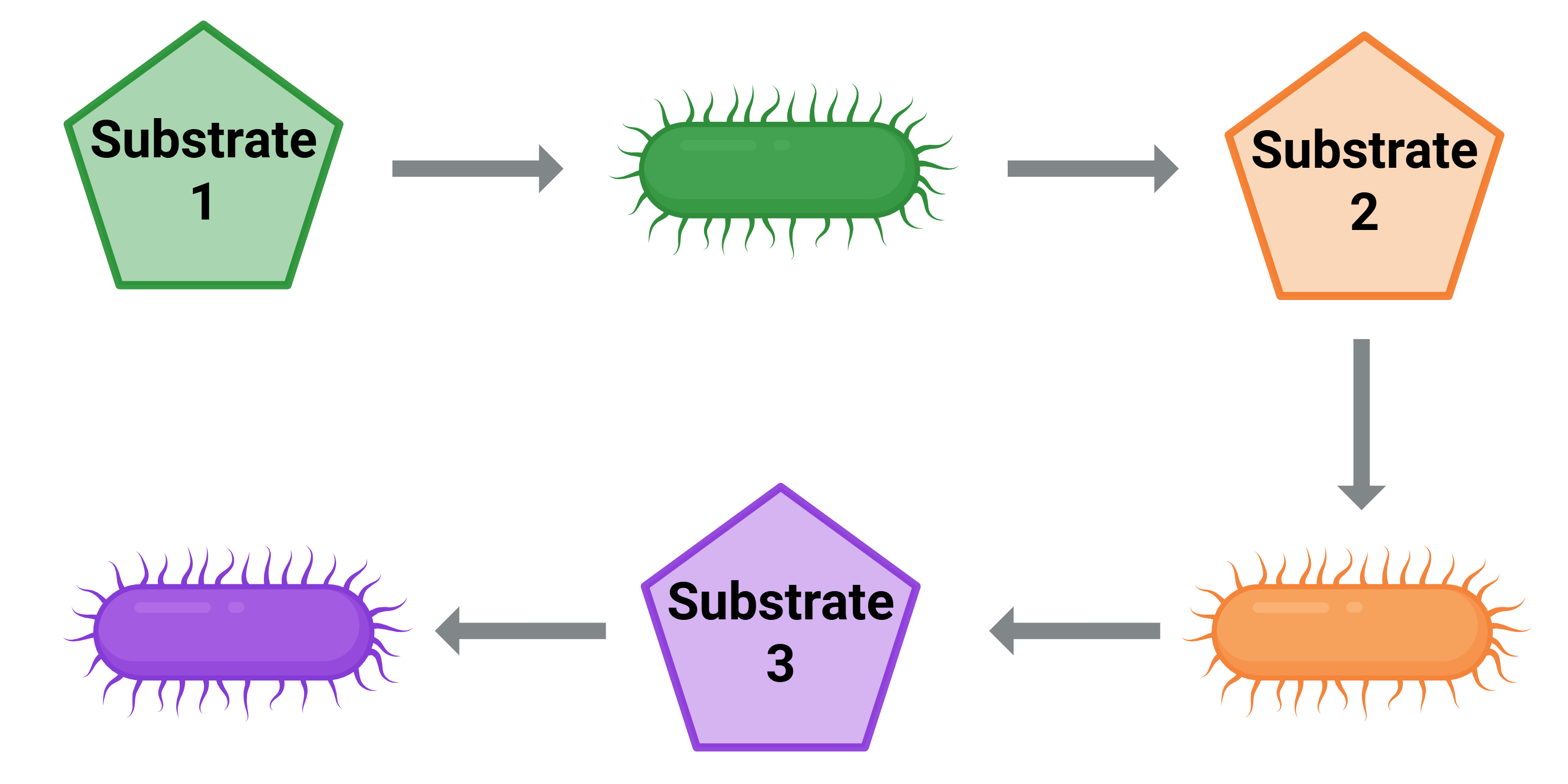}
			\caption{Commensalism}
			\label{fig:schematic_commensalism}
		\end{subfigure}
		\caption{Schematic representations of two interaction types: (a) Competition: The green, orange and purple species directly compete for the consumption of a single substrate; (b) Commensalism: The green species converts substrate 1 into substrate 2, which serves as a resource for the orange species, leading to the production of substrate 3, which is finally consumed by the purple species.}
		\label{fig:schematic}
	\end{figure}
	\begin{itemize}
		\item in the case of competition between three bacterial species for a single substrate (cf. Figure \ref{fig:schematic_competition}), we observe that bacterial species separate into radially growing sectors, \emph{columned stratification}, as has been found in \cite{MASSG22};
		\item in the case of two bacterial species with different growth strategies competing for a single substrate, our results suggest that a \emph{yield strategist} (more efficient use of substrate)  can out-compete a \emph{growth strategist} (higher maximum growth rate), under certain substrate-limited conditions, as in \cite{K04};
		\item in the case of commensalism between three bacterial species (cf. Figure \ref{fig:schematic_commensalism}), we observe \emph{layered stratification}, i.e., different bacterial populations occupy distinct radial layers within the aggregate, as has been reported in \cite{MASSG22}.
	\end{itemize}
	
    We also provide a mathematical analysis of a specific subcase of the competition system, consisting of coupled PDEs for the time evolution of the total biomass and a single substrate densities. We show that this system is well-posed locally in time using Amann's theory for quasilinear parabolic equations, \cite{Am90}. Then, we formally derive a formula for the minimal travelling wave speed governing the expansion of the total biomass, and compare it with the numerically computed wave speed.
	
	\paragraph{Structure of the paper.} In Section~\ref{sec:model}, we list the main components constituting our model, introduce the general PDE system for an arbitrary number of bacterial species and substrates, and define specific cases for subsequent comparison with the IBMs developed previously for these specific cases. In Section~\ref{sec:analysis}, we provide a formal derivation of a pure competition model for $n$ species from an on-lattice biased random walk, a local-in-time well-posedness of a subcase of the latter, and a formal travelling wave analysis of the subcase. We present our numerical results for the selected subcases in Section~\ref{sec:results}, and a final discussion as well as future perspectives in Section~\ref{sec:discussion}. Details on the numerical scheme can be found in Appendix~\ref{app:numericalscheme}.
	
	\section{The PDE framework} \label{sec:model}
	
	In this section, we outline our main modelling assumptions, formulate the general framework for an arbitrary number of bacterial species and substrates, and then introduce three specific cases: pure competition, pure commensalism, and mixed interactions for a system of nitrifying bacteria. Finally, we briefly comment on the difference of our models compared to similar PDE ones.
    \paragraph {Movement of the bacteria.} The movement of bacteria is modelled as a combination of diffusion and advection down a pressure gradient. Diffusion is assumed to be negligible compared to advection, reflecting the dominance of pressure-driven shoving effects in densely packed aggregates. The pressure is assumed to depend linearly on the total mass density of the microbial population, which represents the mechanical crowding and displacement forces observed in the individual-based model \cite{MASSG22}.

    \paragraph{Growth and decay of bacterial populations.} We assume that bacterial growth is limited by the availability of some substrates and follows Monod kinetics. This approach reflects the saturable nature of microbial metabolism in nutrient-limited environments. Competition can be introduced between several populations by the shared use of one substrate. In addition, a density-dependent death term is included to account for the effects of crowding and resource depletion, which can lead to cell mortality as population density increases.
    
    \paragraph{Diffusion of substrates.} Substrate transport is modelled by standard Fickian diffusion. Nutrients may enter the system through the domain boundary, but the framework can be adapted to different boundary conditions, depending on how nutrient supply is implemented in the experimental setting. Mathematically, this is typically represented by Neumann conditions for prescribed fluxes, Dirichlet conditions for fixed concentrations at the boundary, or a combination of these.
    
    \paragraph{Consumption and/or production of substrates.} Certain substrates may function as nutrients for specific microbial populations, with uptake described by Monod kinetics and scaled by the corresponding yield coefficients. In addition, substrates may be generated as metabolic by-products, whereby particular bacterial species transform one substrate into another.

    \smallskip
	Based on these assumptions, we present a general framework for $n$ bacterial species with $m$ substrates in two spatial dimensions. 
	\subsection{A general model}  \label{sec:generalPDE}
	
	Let $u_i(x,t)$ denote the population density of the $i$-th species, for $1 \leq i \leq n $ and $c_j(x,t)$ be the concentration of the $j$-th substrate, for $1\leq j \leq m$, where $x \in \Omega \subset \R^2, \, t >0$. A similar derivation to that in Section~\ref{sec:derivation} leads to the following system of PDEs, for $1\leq i \leq n $, $1\leq j \leq m$,
	\begin{align}
		\label{eq:generalPDEmodel}
		\begin{split}
			\frac{\partial }{\partial t} u_i &= \nabla \cdot \left[ d_i \nabla u_i + a_i u_i \nabla f(\rho) \right] + u_i\left( r_{i} \prod_{j=1}^m  \frac{c_j}{K_{ij} + c_j } - b_i \rho  \right), 
			\\
			\frac{\partial }{\partial t} c_j &= D_j \Delta c_j - \sum_{i=1}^n  u_i \delta_{ij} \frac{r_{i}}{Y_{ij}} \prod_{k=1}^m \frac{c_k} {K_{ik} + c_k} + \sum_{i=1}^n \sigma_{ij} u_i r_i  \prod_{k=1, k \neq j}^m \frac{c_k} {K_{ik} + c_k},
		\end{split}
	\end{align} where $\rho (x,t) := \sum_i^n u_i(x,t)$ is the total bacterial density. The diffusion coefficient for the $j$-th substrate and the diffusion and advection coefficients for the $i$-th species are referred to as  $D_j>0$, $d_i > 0$, and $a_i >0$, respectively. The function $f$ represents the internal pressure and is assumed to be a smooth function of the density $\rho$, monotonically increasing for $\rho >0$ and $f(0) = 0$. Although one could consider various density-dependent pressure terms, see, e.g., \cite{CLM20, CLM25}, for simplicity, we will assume $f(\rho):=\rho$ for the remainder of this article. If crowding-induced displacement is the primary driver of spatial displacement, a natural assumption is $d_i \ll a_i$ for $1 \leq i \leq n$. Moreover, we denote the death rates by $b_i>0$, the maximal growth rates by $r_i>0$ with half-saturation concentrations $K_{ij}$ and yield coefficients $Y_{ij}$ in the Monod term. For the coefficients involved in the Monod term for nutrient uptake, we make a case distinction: 
	\begin{itemize}
		\item If the $j$-th substrate is consumed as a nutrient by the $i$-th bacterial species, then $\delta_{ij} = 1$ and we assume that both the half-saturation concentration $K_{ij}>0$ and the yield coefficient $Y_{ij}>0$ are positive constants.
		\item Otherwise, we have $\delta_{ij}=0$, $K_{ij}=0$ and $Y_{ij}$ does not need to be defined.
	\end{itemize}
	Note that in the latter case, the presence of substrate $c_j$ does not limit the growth of population $u_i$ and the artificial loss of substrate $c_j$ is avoided by setting $\delta_{ij}=0$. If the $j$-th substrate is a metabolic product of species $i$, then the conversion factor $\sigma_{ij}>0$. Otherwise, we set $\sigma_{ij}=0$ to avoid the artificial creation of substrate. Note also that, for biological realism, a single species should not both consume and produce the same substrate; that is, for fixed $i,j$ at most one of $\delta_{ij}$ or $\sigma_{ij}$ may be positive.  
	
	System \eqref{eq:generalPDEmodel} is equipped with initial and Neumann-type boundary conditions,
	\begin{align} \label{eq:generalPDE_IC_BC}
		\begin{split}
			u_i (x,0) &= u^0_i (x)  \, \, \text{ for } x \in \Omega, \qquad \qquad \frac{\partial}{\partial \nu} u_i = 0  \quad \text{  on } \partial \Omega,  \quad 1\leq i \leq n, \\
			c_j(x,0)   &= c_j^0 \quad \, \, \,\, \, \text{ for } x \in \Omega, \qquad   \quad \quad \frac{\partial}{\partial \nu} c_j  = c_j^\infty \, \,  \text{  on } \partial \Omega, \quad 1 \leq j \leq m,
		\end{split}
	\end{align}
	where $\nu$ is the outer normal vector on $\partial \Omega$, $c_j^\infty$, $c_j^0$ are non-negative constants for $1\leq j\leq m$, and $u^0_i (x)$ are non-negative functions for $1 \leq i \leq n$.

	\paragraph{Comparison of our system \eqref{eq:generalPDEmodel} with similar models in the literature.} The flux term for the bacterial movement in \eqref{eq:generalPDEmodel} contains a crowding- or pressure-driven cross-diffusion mechanism of the form $u_i\nabla f(\rho)$, which appears in various continuum frameworks for microbial or cellular populations.  Classical biofilm models often use density-dependent diffusion (sometimes degenerate), e.g. $\nabla\!\cdot(D(u)\nabla u)$, to capture biomass spreading and volume filling, as in \cite{KD02, SEE15, HS22}. More general multispecies cross-diffusion formulations have been proposed for mixed biofilms, extending these ideas to interacting populations \cite{AK07, RSRW15}. In tumour models, cell velocity is frequently assumed to follow Darcy’s law, i.e., to be proportional to the pressure gradient, leading to fluxes of the form $u\nabla p(\rho)$, where the pressure $p(\rho)$ is a function of the cell density, coupled to reaction kinetics, as seen in \cite{CLM20,NS16}. Compared to these PDE models, our model couples a pressure-driven movement law to multi-substrate kinetics. Our approach combines crowding-induced passive movement with metabolic interactions in a general $n$-species, $m$-substrate framework. 
	
	\subsection{Special cases} \label{sec:cases}
	
	In this section, we present how the general model \eqref{eq:generalPDEmodel} applies to three cases studied in the literature via IBMs \cite{K04,  MSSG23,MASSG22}: (i) competition, (ii) commensalism, and (iii) a mixed interaction system involving both competition and commensalism for a system of nitrifying bacteria. 
	
	\paragraph{Competition.} In this case, microorganisms compete for the same substrate (cf. Figure \ref{fig:schematic_competition}). This can lead to competitive exclusion, a central principle in microbial ecology  \cite{PF22}, where two species occupying similar niches cannot coexist and one will eventually exclude the other \cite{FR12}. 
	
	Let $n \geq 2$ bacterial species with densities $u_i(x,t), \, 1 \leq i \leq n$ compete for a single nutrient with concentration $c(x,t)$ where $(x,t) \in \Omega \times [0, +\infty), \, \Omega \in \R^2$. Then, the system \eqref{eq:generalPDEmodel} can be simplified to
	\begin{align}
		\label{eq:comp3}
		\begin{split}
			\frac{\partial }{\partial t} u_i &= \nabla \cdot \left[ d_i \nabla u_i + a_i u_i \nabla \rho \right] + u_i \left( r_i \frac{c}{K_i + c} - b_i \rho \right), \quad 1 \leq i \leq n , \\
			\frac{\partial }{\partial t} c &= D \Delta c - \sum_{i=1}^n u_i \frac{r_i}{Y_i} \frac{c}{K_i + c},
		\end{split}
	\end{align}
	where $\rho (x,t):= \sum_{i=1}^n u_i (x,t)$ is the total biomass density.
	System \eqref{eq:comp3} is complemented with the following initial and Neumann boundary conditions,
	\begin{alignat*} {4}
		u_i (x,0) &= u^0_i (x)  \quad &&\text{ for } x \in \Omega, \qquad \qquad &&\frac{\partial}{\partial \nu} u_i = 0  \quad  \qquad \, &&\text{  on } \partial \Omega, \quad 1 \leq i \leq n,\\
		c(x,0)   &= c^0 (x, 0)>0 &&\text{ for } x \in \Omega, \qquad &&\frac{\partial}{\partial \nu} c  = c^\infty >0  &&\text{  on } \partial \Omega.
	\end{alignat*}
    In the competition case, the IBMs typically produce columned stratification \cite{MASSG22, WHBI25, BL15}, with species segregating into radially oriented patches controlled by competition for space and access to the limiting substrate, e.g., \cite{MASSG22} for three bacterial species. This columned spatial structure allows less competitive populations to persist alongside faster-growing competitors, because the geometry ensures that even low-fitness cells retain access to the limiting resource. In environments with spatially varying substrate concentrations, bacteria that have a higher yield but grow more slowly when substrate is abundant appear to have a competitive advantage when substrate is scarce \cite{K04}.
    
    In Section \ref{sec:results}, we investigate System \eqref{eq:comp3} for $n=2$ and $n=3$ species numerically, and compare our results with the IBMs in \cite{K04, MASSG22}. Furthermore, in Section \ref{sec:derivation} we formally derive \eqref{eq:comp3} in one spatial dimension from an on-lattice biased random walk. For a parameter-symmetric version of \eqref{eq:comp3} we provide a local-in-time well-posedness result and a formal travelling wave analysis in Sections~\ref{sec:well_posedness} and \ref{sec:travelling_wave}, respectively.

	\paragraph{Commensalism.} In this setting, one organism benefits from another without affecting it in return. Such interactions are commonly found in biodegradation, in which some microorganisms cross-feed on by-products that are created by others \cite{FR12}. 
	
	We consider cooperating interactions between three microbial species with densities $u_i(x,t)$, $1\leq i \leq 3$, and three substrates with concentrations $c_j(x,t)$, $1\leq j \leq 3$, to facilitate comparison with \cite{MASSG22}. That is, Species $1$ consumes the externally supplied Substrate $1$ and produces Substrate $2$, which serves as nutrient for Species $2$. In turn, Species $2$ produces Substrate $3$, which is ultimately consumed by Species $3$ (cf. Figure~\ref{fig:schematic_commensalism}).
	
	In this case, System \eqref{eq:generalPDEmodel} takes the form,
	\begin{align}
		\label{eq:comm}
		\begin{split}
			\frac{\partial }{\partial t} u_i &= \nabla \cdot \left[ d_i \nabla u_i + a_i u_i \nabla \rho \right] + u_i \left(r_i \frac{c_i}{K_i + c_i} - b_i \rho  \right), \quad 1 \leq i  \leq 3,\\
			\frac{\partial }{\partial t} c_1 &=D_1 \Delta c_1 - u_1 \frac{r_1}{Y_1} \frac{c_1}{K_1 + c_1},\\
			\frac{\partial }{\partial t} c_2 &= D_2 \Delta c_2 - u_2 \frac{r_2}{Y_2} \frac{c_2}{K_2 + c_2} +\sigma_{12} u_1 \frac{r_1}{Y_1} \frac{c_1}{K_1 + c_1}, \\
			\frac{\partial }{\partial t} c_3 &= D_3 \Delta c_3 - u_3 \frac{r_3}{Y_3} \frac{c_3}{K_3 + c_3} +\sigma_{23} u_2 \frac{r_2}{Y_2} \frac{c_2}{K_2 + c_2},
		\end{split}
	\end{align} where $\rho (x,t) := \sum_{i=1}^3 u_i (x,t)$ is the total biomass density,  $ \sigma_{12}, \sigma_{23} >0$ are the metabolic substrate conversion rates from $c_1 \to c_2$ and $c_2 \to c_3$, respectively. We complement \eqref{eq:comm} with the initial and Neumann boundary conditions given by
	\begin{alignat*}{4}
		u_i (x,0) &= u^0_i (x)  &&\text{for } x \in \Omega, \qquad \qquad&&\frac{\partial}{\partial \nu} u_i = 0  \quad &&\text{  on } \partial \Omega, \quad 1\leq i \leq 3, \\
		c_1(x,t)  &= c_1^0 >0\quad   &&\text{for } x \in \Omega, \qquad   &&\frac{\partial}{\partial \nu} c_1  = c_1^\infty >0 \, &&\text{  on } \partial \Omega,\\
		c_j(x,t)   &= 0 &&\text{for } x \in \Omega, \qquad   &&\frac{\partial}{\partial \nu} c_j  = 0  &&\text{  on } \partial \Omega, \quad 2 \leq j \leq 3.
	\end{alignat*}
	
    In commensal scenarios, the IBM in \cite{MASSG22} exhibits a spatial structure referred to as \textit{layered stratification}: populations form concentric radial layers within the aggregate, organised according to the order in which substrates are consumed.
    In Section \ref{sec:results}, we study \eqref{eq:comm} numerically and compare our results with the IBM studied in \cite{MASSG22}. 
	
	\paragraph{Mixed ecological interactions involved in nitrification.} 
	For our final scenario, we consider a well-studied system of nitrifying bacteria, which is characterised by both commensal and competitive interactions. Nitrification has long been believed to be a two-step process in which ammonia-oxidising bacteria (AOB) convert ammonium ($\ce{NH_4^+}$) into nitrite ($\ce{NO_2^-}$), and subsequently, nitrite-oxidising bacteria (NOB) convert nitrite ($\ce{NO_2^-}$) into nitrate ($\ce{NO_3^-}$), see, e.g., \cite{SK16}. Today it is known that certain species perform complete ammonia oxidation, directly converting ammonium ($\ce{NH_4^+}$) into nitrate ($\ce{NO_3^-}$), e.g., \cite{KSANOKJL15, DLP15}. These species are referred to as comammox (CMX) organisms. The stoichiometric relations governing these three oxidation processes are given by 
	\begin{align}
		\ce{ NH_4 ^+ + 1.5 O_2  &-> NO_2^- + 2H ^+ + H_2 O }, \tag{AOB}\\
		\ce{NO^-_2 + 0.5 O_2 &-> NO^-_3}, \tag{NOB}\\
		\ce{NH^+_4 + 2O_2 &-> NO_3^- + 2H^+ + H_2O}.  \tag{CMX}
	\end{align}  
	While AOB and NOB form a partially commensal relationship, AOB and CMX compete directly for ammonium $\ce{NH_4^+}$. Moreover, oxygen ($\ce{O_2}$) is required for all three of those processes, putting all three of the species into direct competition for $\ce{O_2}$ (cf. Figure~\ref{fig:schematic_nitrification}).
    	\begin{figure}[ht!]
		\centering
		\includegraphics[width=0.65\textwidth]{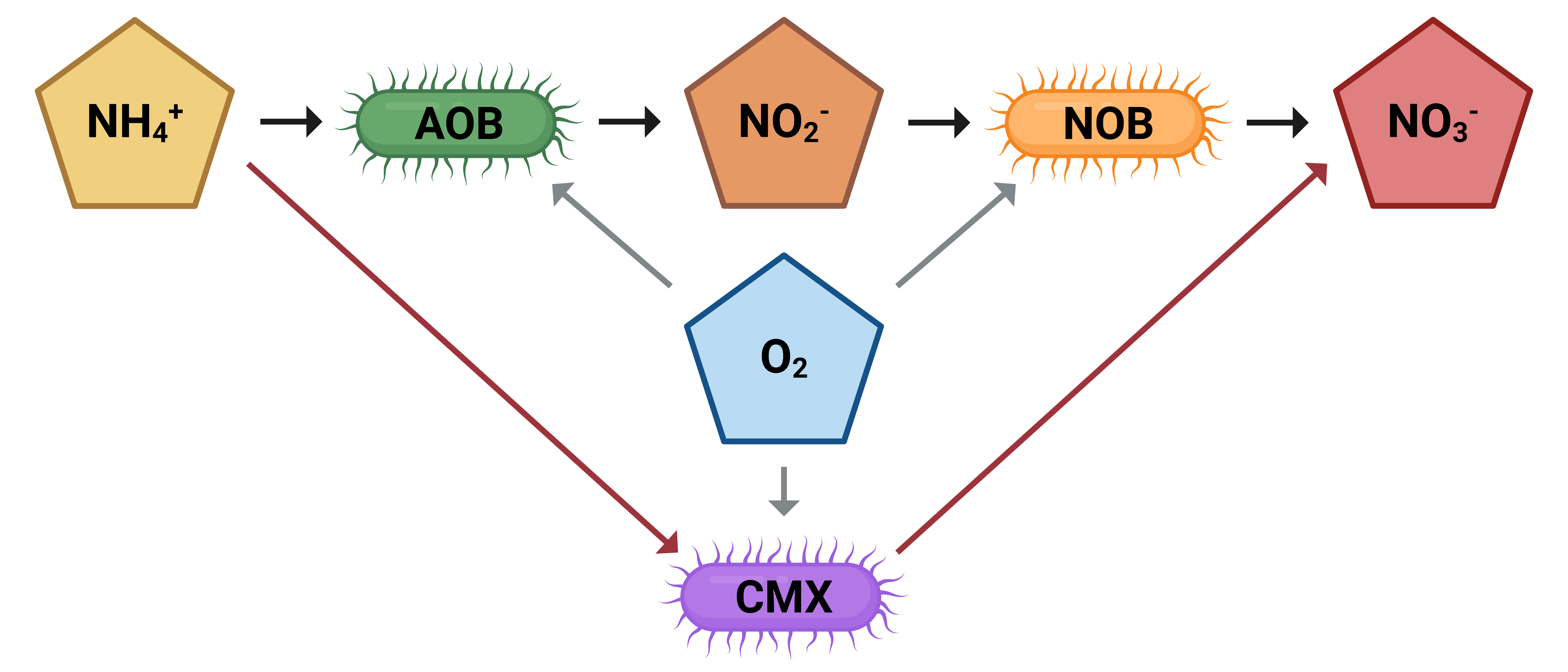}
		\caption{Two paths of nitrification: a two-step process performed by AOB and NOB (black top path), and complete ammonia oxidation by CMX (red bottom path). Ecological interactions include a commensal relation between AOB and NOB, competition between AOB and CMX for $\ce{NH_4^+}$, and finally, competition among all three species for $\ce{O_2}$.}
		\label{fig:schematic_nitrification}
	\end{figure}

        Applied to this example, System \eqref{eq:generalPDEmodel} comprises three equations for the bacterial species AOB, NOB, and CMX with densities $u_i(x,t)$, $1 \leq i \leq 3$, respectively, and four equations for the substrates ammomium $\ce{NH_4^+}$, nitrite $\ce{NO_2^-}$, nitrate $\ce{NO_3^-}$ and oxygen $\ce{O_2}$ with concentrations $c_j(x,t)$, $1\leq j \leq 4$, respectively. The model is then given by 
	\begin{align}
		\label{eq:mixed_interaction}
		\frac{\partial }{\partial t} u_1 &= \nabla \cdot \left[ d_1 \nabla u_1 + a_1 u_1 \nabla \rho \right] + u_1 \left(r_1 \frac{c_1}{K_{11} + c_1} \frac{c_4}{K_{14} + c_4} - b_1 \rho \right), \nonumber \\
		\frac{\partial }{\partial t} u_2 &= \nabla \cdot \left[ d_2 \nabla u_2 + a_2 u_2 \nabla \rho \right] + u_2 \left(r_2 \frac{c_2}{K_{22} + c_2} \frac{c_4}{K_{24} + c_4} - b_2 \rho \right), \nonumber\\
		\frac{\partial }{\partial t} u_3 &= \nabla \cdot \left[ d_3 \nabla u_3 + a_3 u_3 \nabla \rho \right] + u_3 \left(r_3 \frac{c_1}{K_{31} + c_1} \frac{c_4}{K_{34} + c_4} - b_3 \rho \right), \nonumber\\
		\frac{\partial }{\partial t} c_1 &=D_1 \Delta c_1 - u_1 \frac{r_1}{Y_{11}} \frac{c_1}{K_{11} + c_1}  \frac{c_4}{K_{14} + c_4} - u_3 \frac{r_3}{Y_{31}} \frac{c_1}{K_{31} + c_1} \frac{c_4}{K_{34} + c_4},  \tag{2.5}\\
		\frac{\partial }{\partial t} c_2 &= D_2 \Delta c_2 + \sigma_{12} u_1 r_1 \frac{c_1}{K_{11} + c_1} \frac{c_4}{K_{14} + c_4}-  u_2 \frac{r_2}{Y_{22}} \frac{c_2}{K_{22} + c_2} \frac{c_4}{K_{24} + c_4} ,  \nonumber\\
		\frac{\partial }{\partial t} c_3 &= D_3 \Delta c_3  +\sigma_{23} u_2 r_2 \frac{c_2}{K_{22} + c_2} \frac{c_4}{K_{24} + c_4} + \sigma_{13} u_3 r_3 \frac{c_1}{K_{31} + c_1} \frac{c_4}{K_{34} + c_4}, \nonumber\\
		\frac{\partial }{\partial t} c_4 &= D_4 \Delta c_4 - u_1 \frac{r_1}{Y_{14}} \frac{c_1}{K_{11} + c_1} \frac{c_4}{K_{14} + c_4} - u_2 \frac{r_2}{Y_{24}} \frac{c_2}{K_{22} + c_2} \frac{c_4}{K_{24} + c_4}- u_3 \frac{r_3}{Y_{34}} \frac{c_1}{K_{31} + c_1} \frac{c_4}{K_{34} + c_4}, \nonumber
	\end{align} 
	where $\rho (x,t) := \sum_{i=1,2,3} u_i (x,t)$ is the total biomass density and $\sigma_{12}, \sigma_{23}, \sigma_{13} >0$ are the metabolic substrate conversion rates from $c_1 \to c_2$, $c_2 \to c_3$ and $c_1 \to c_3$, respectively. The system is also equipped with initial and Neumann boundary conditions given by,
	\begin{alignat*}{4}
		u_i (x,0) &= u^0_i (x)  &&\text{for } x \in \Omega, \qquad \qquad&&\frac{\partial}{\partial \nu} u_i = 0  \quad &&\text{  on } \partial \Omega,  \quad 1 \leq i \leq 3,\\
		c_1(x,t)  &= c_1^0 >0\quad   &&\text{for } x \in \Omega, \qquad   &&\frac{\partial}{\partial \nu} c_1  = c_1^\infty >0 \, &&\text{  on } \partial \Omega,\\
		c_j(x,t)   &= 0 &&\text{for } x \in \Omega, \qquad   &&\frac{\partial}{\partial \nu} c_j  = 0  &&\text{  on } \partial \Omega, \quad 2 \leq j \leq 3,\\
		c_4(x,t)   &=  c_4^0 >0\quad &&\text{for } x \in \Omega, \qquad   &&\frac{\partial}{\partial \nu} c_4 = c_4^\infty >0  &&\text{  on } \partial \Omega.
	\end{alignat*} The boundary conditions above imply that ammonium $c_1$ and oxygen $c_4$ are provided externally, while both the nitrite $c_2, \, c_3$ are metabolic by-products.

	Even though we apply our model to this scenario, we reserve its detailed analysis for future study. The realistic parameter values operate on different time scales, suggesting that a time-scale separation should be incorporated into the numerical scheme.  Accommodating this in our current implementation would require nontrivial modifications. Moreover, the expected qualitative outcomes of the model depend sensitively on the specific parameter regimes and ecological configurations considered, making this case a natural subject for a separate, dedicated study.

	\section{Theoretical results} \label{sec:analysis} 
	In this section, we formally derive the competition system \eqref{eq:comp3} in one spatial dimension from an on-lattice biased random walk. Subsequently, we consider a parameter-symmetric case the competition system and provide a local-in-time well-posedness result as well as a formal travelling wave analysis.

	\subsection{Derivation from a lattice model} \label{sec:derivation}
	
	We formally derive the PDE system \eqref{eq:comp3} for $n$ bacterial species competing for one substrate from an on-lattice biased random walk in one spatial dimension. This derivation also serves as a justification for the specific form of the movement terms that we consider for bacteria motility. We employ standard techniques for deriving a PDE system from a biased random walk, following approaches used in, for example, \cite{OS97}. In particular, we draw attention to \cite{CLM20, CLM25}, where a similar pressure-gradient was incorporated into the derivation.

	We discretise the time $t \in \R^+$ and space $x\in \R $ variables as $t_k = k \Delta t$, $k \in \mathbb{N}$ and $x_l = l \Delta x$, $ l \in \mathbb{Z} $, respectively, where $0 <\Delta t, \Delta x \ll 1$ are the discretisation steps in time and space, respectively. We introduce the variables $N_i(x_l,t_k) \in \mathbb{N}$ and $S(x_l,t_k)\in \mathbb{N}$ to model the number of bacteria in the $i$-th population and the number of substrate units at $x_l$ at time $t_k$, respectively. Then, the density of the $i$-th population $u_i(x_l, t_k)$, the density of the total bacterial population $\rho(x_l, t_k)$ and the concentration of the substrate $c(x_l, t_k)$ can be computed as, for $1 \leq i \leq n$,
	\begin{align*}
		u_i(x_l,t_k) = u_{il}^k =\frac{N_i(x_l,t_k)}{\Delta x}, \qquad \rho(x_l,t_k) = \rho_{l}^k= \sum_{i=1}^n u_i(x_l,t_k), \qquad c(x_l,t_k) = c_l^k  = \frac{S(x_l,t_k)}{\Delta x},
	\end{align*}
	We assume that, at each time step, every individual has a certain probability of moving and it either proliferates or is consumed, in line with the on-lattice random-walk rules that we formulate below.
	\paragraph{Movement of bacteria.} Individual bacteria may move in both directions on the one-dimensional lattice. We model the tendency of being shoved away through crowding as increased probability to actively move away from higher population densities. 
	We define the probability of species $i$ moving from a position $x_l$ to an adjacent point on the grid $x_{l \pm 1}$ at a time $t_k$, for $1\leq i \leq n $, by
	\begin{equation} 
		\label{eq:prob_jumpU}
		\mathcal{P}_i (x_l \to x_{l\pm1})= \tilde{d}_i + \tilde{a}_i\lp  \rho (x_l, t_k) - \rho(x_{l\pm 1}, t_k)\rp_+, 
	\end{equation}
	with $\tilde{d}_i,\tilde{a}_i>0$, $(\cdot)_+ := \max \{\cdot, 0\}$.
	The first constant $\tilde{d}_i$ of the movement probability corresponds to purely random movement along the grid, while the second term reflects the bias along the population density gradient with scaling $\tilde{a}_i$.
	The probability of not jumping and remaining at position $x_l$ at time $t_k$ then becomes 
	\begin{equation}
		\label{eq:prob_notmoveU}
		\mc P_i^{\text{stay}} = 1 - \mc P_i (x_l \to x_{l+1}) - \mc P_i (x_l \to x_{l-1}). 
	\end{equation}
	
	\paragraph{Growth and decay of bacteria.} All bacterial species rely on the consumption of substrate to grow. The probability of bacteria of species $i$ proliferating within a time interval $\Delta t$ thus depends on the local concentration of the substrate $c$, where we assume growth to follow Monod kinetics, i.e.,
	\begin{equation}
		\label{eq:prob_growU}
		\mathcal{P}_i^{\text{grow}} = \Delta t r_i \frac{c(x_l,t_k)}{K_i + c(x_l,t_k)},
	\end{equation}
	with maximum growth rate $r_i>0$ and half-saturation concentration $K_i>0$. We furthermore assume that the probability of decaying within a timestep $\Delta t$ is density-dependent, where higher total population densities corresponds to a higher decay probability, i.e.,
	\begin{equation}
		\label{eq:prob_decayU}
		\mathcal{P}_i^{\text{decay}} = \Delta t b_i \rho(x_l,t_k),
	\end{equation}
	with decay rate $b_i>0$. Bacteria can also stay quiescent, i.e., neither grow nor decay with probability,
	\begin{equation}
		\label{eq:prob_quiescU}
		\mathcal{P}_i^{\text{quiescence}} = 1 - \mathcal{P}_i^{\text{decay}} - \mathcal{P}_i^{\text{grow}}.
	\end{equation}
	
	\paragraph{Diffusion of the substrate.}  The chemical substrate diffuses freely in the domain, which is represented by an unbiased random walk on the grid, i.e., the probability of a substrate moving from $x_l$ to an adjacent position $x_{l \pm 1}$ is given by a constant
	\begin{equation}
		\label{eq:prob_diffusionC}
		\mathcal{Q}^{\text{move}} = \tilde{D},
	\end{equation}
	where $0< \tilde{D}\leq 1/2$. The probability for the substrate not moving is then
	\begin{equation}
		\label{eq:prob_notmovingC}
		\mathcal{Q}^{\text{stay}} = 1- 2 \tilde{D}.
	\end{equation}
	
	\paragraph{Consumption of the substrate.} The substrate serves as growth-limiting nutrient and is consumed during bacterial proliferation. Hence, we couple the probability of consumption per time $\Delta t$ to the proliferation probabilities and scale them by species-specific yield coefficients in the following way,
	\begin{equation}
		\label{eq:prob_consumptionC}
		\mathcal{Q}^{\text{consumed}} = \Delta t \sum_{i=1}^n \frac{r_i}{Y_i} \frac{c(x_l, t_k)}{K_i + c(x_l, t_k)},
	\end{equation}
	where $Y_i>0$ is the yield coefficient (produced biomass per unit substrate) for the $i$-th species.

	Finally, to ensure that all probabilities \eqref{eq:prob_jumpU}-\eqref{eq:prob_consumptionC} lie in the interval $[0,1]$, we choose the time step $\Delta t$ and the jump parameters $\tilde{d}_i$, $\tilde{a}_i$ sufficiently small. Then, the master equations for the $i$-th species, $1\leq i \leq n $, reads as follows, 
	\begin{multline}
		\label{eq:masterP}
		u_i(x_l,t_{k+1}) - u_i(x_l,t_k) \\ =  \sum_{y = x_{l\pm 1 }}   \left[ \mathcal{P}_i (y \to x_l) u_i(y,t_k) - \mathcal{P}_i (x_l \to y) u_i(x_l,t_k) \right] 	+ \lp \mathcal{P}_i^{\text{grow}} - \mathcal{P}_i^{ \text{decay}} \rp  u_i(x_l,t_k),
	\end{multline} where the first term on the right hand side comes from the biased random walk, while the second one accounts for the net growth. Similarly, the master equation for the substrate $c$ is given by
	\begin{equation*}
		c(x_l, t_{k+1}) - c (x_l,t_k) =  \mathcal{Q}^{\text{move}} \lp \sum_{y = x_{l \pm 1}} c(y,t_k) - 2c(x_l,t_k) \rp  - \mathcal{Q}^{ \text{consumed}}c(x_l,t_k).
	\end{equation*}
	In the remainder of this section, we will use the following, since $\Delta t, \Delta x$ are chosen to be sufficently small, for $1 \leq i \leq n$,
	\begin{align*}
		x_l \approx x, \quad x_{l \pm 1 } \approx x &\pm \Delta x, \quad t_k \approx t, \quad t_{k \pm 1 } \approx t \pm \Delta t, \\ u_i(x_l,t_k) \approx u_i(x,t), \quad &c(x_l,t_k) \approx c(x,t), \quad \rho(x_l,t_k) \approx \rho(x,t).
	\end{align*}
	First, we demonstrate the derivation of a continuum equation for the density of the $i$-th species. We start by substituting the probabilities \eqref{eq:prob_jumpU}-\eqref{eq:prob_quiescU} into \eqref{eq:masterP}, and we obtain
	\begin{equation*}
		\begin{split}
			u_i (x, t+ \Delta t) - u_i (x,t) =& \quad  \lp  \tilde{d}_i +  \tilde{a}_i  \lp  \rho (x+\Delta x, t) - \rho (x,t) \rp_+ \rp u_i (x+ \Delta x, t) \\
			& +  \left( \tilde{d}_i +  \tilde{a}_i   \lp  \rho (x - \Delta x, t) - \rho (x,t) \rp_+\right) u_i (x- \Delta x, t) \\
			& - u_i (x,t)  \left( 2\tilde{d}_i +  \tilde{a}_i   \left[  \lp \rho (x,t)  - \rho (x + \Delta x, t) \rp_+ + \lp \rho (x,t )- \rho (x - \Delta x, t) \rp_+\right] \right) \\
			& + u_i(x,t) \Delta t \left( r_i \frac{c(x, t)}{K_i + c(x, t)}  -  b_i \rho(x,t) \right).
		\end{split}
	\end{equation*}
	Assuming that $u_i(x,t)$, $1\leq i \leq n$, and $c(x,t)$ are sufficiently smooth with respect to $x$ and $t$, we can employ Taylor expansions of $\rho (x \pm \Delta x, t), u_i(x \pm \Delta x, t)$ about $(x,t)$ up to the second order in $x$ and to first order in $t$. 
Making use of $\lp \star \rp_+  - \lp - \star \rp_+ = \star$, above expression simplifies to
	\begin{align*}
		\Delta t \, \frac{\p}{\p t} &u_i(x,t)  + \mc O  (\Delta t^2) =  \tilde d_i \lp \Delta x^2 \frac{\p^2}{\p x^2} u_i(x,t) + \mc O (\Delta x^3) \rp  + \tilde a_i \lp \Delta x^2 \frac{\p^2}{\p x^2} \rho (x,t) + \mc O (\Delta x^3)\rp u_i (x,t) \\
		&+ \tilde a_i \lp \Delta x \frac{\p}{\p x} \rho (x,t) + \frac{\Delta x ^2}{2} \frac{\p^2}{\p ^2x }\rho (x,t) \rp_+ \lp \Delta x \frac{\p}{\p x} u_i(x,t) + \frac{\Delta x^2}{2} \frac{\p^2}{\p x^2} u_i(x,t) + \mc O (\Delta x^3)\rp \\
		&+ \tilde a_i \lp -\Delta x \frac{\p}{\p x} \rho (x,t) + \frac{\Delta x ^2}{2} \frac{\p^2}{\p x^2 }\rho (x,t)   \rp_+ \lp - \Delta x \frac{\p}{\p x} u_i(x,t) + \frac{\Delta x^2}{2} \frac{\p^2}{\p x^2} u_i(x,t) + \mc O (\Delta x^3)\rp\\
		&+ u_i(x,t) \Delta t \left( r_i \frac{c(x, t)}{K_i + c(x, t)}  -  b_i \rho(x,t) \right).
	\end{align*} 
	A further simplification yields
	\begin{align*}
		\Delta t \, \frac{\p}{\p t} u_i(x,t)  &=  \tilde d_i  \Delta x^2 \frac{\p^2}{\p x^2} u_i(x,t)  + \tilde a_i \Delta x^2 \frac{\p^2}{\p x^2} \rho (x,t)  u_i (x,t) 
		+ \tilde a_i  \Delta x^2 \frac{\p}{\p x} \rho(x,t)  \frac{\p}{\p x} u_i(x,t) \\
		&+ u_i(x,t) \Delta t \left( r_i \frac{c(x, t)}{K_i + c(x, t)}  -  b_i \rho(x,t) \right) + \mc O (\Delta x^3)  +\mc O (\Delta t^2).
	\end{align*} 
	Dividing both sides by $\Delta t$, and rearranging the terms give
	\begin{align*}
		\frac{\p}{\p t} u_i(x,t) =   &\frac{\Delta x^2}{\Delta t} \left [  \tilde{d}_i \frac{\p^2}{\p x^2} u_i(x,t)  + \tilde a_i  \lp u_i(x,t) \frac{\p^2}{\p x^2} \rho (x,t)+\frac{\p}{\p x} u_i(x,t)  \frac{\p}{\p x} \rho (x,t)\rp \right ]  
		\\&+
		u_i(x,t)  \lp  r_i \frac{c(x, t)}{K_i + c(x, t)}  -  b_i \rho(x,t)  \rp  + \mc O \lp \frac{\Delta x^3}{\Delta t}\rp + \mc O (\Delta t).
	\end{align*} 
	Letting $\Delta x\to 0$, $\Delta t\to 0$, we formally take the limit such that
	\begin{align*}
		\lim_{\Delta x, \Delta t \to 0} \frac{\Delta x^2}{\Delta t} \tilde{d}_i := 	d_i,  \qquad  \lim_{\Delta x, \Delta t \to 0} \frac{\Delta x^2}{\Delta t} \tilde{a}_i :=  a_i, \quad 1 \leq i \leq n.
	\end{align*} Therefore, we obtain formally the system of PDEs,
	\begin{equation}
		\label{eq:derivation_FinalU}
		\begin{split}
			\partial_t u_i = & \quad \frac{\partial}{\partial x} \left[d_i \frac{\partial}{\partial x} u_i + a_i u_i \frac{\partial}{\partial x}  \rho \right]  + u_i  \left( r_i \frac{c}{K_i + c}  -  b_i \rho \right), \quad 1 \leq i \leq n.
		\end{split}
	\end{equation}
	In a similar manner, we expand the master equation for $c(x,t)$ to obtain, 
	\begin{equation*}
		\begin{split}
			\frac{\partial}{\partial t} c(x,t)  = &  \tilde{D} \frac{\Delta x^2}{\Delta t }  \frac{\partial^2}{\partial x^2} c(x,t) 
			+ \sum_{i=1}^n \frac{r_i}{Y_i} \frac{c(x,t)}{K_i +c(x,t)} + \mathcal{O}\left(\frac{\Delta x^3}{\Delta t}\right) + \mathcal{O}(\Delta t),
		\end{split}
	\end{equation*}
	and, finally, after formally taking the limits $\Delta x\to 0$, $\Delta t\to 0$ such that $\lim_{\Delta x, \Delta t \to 0} \frac{\Delta x^2}{\Delta t} \tilde{D}:= D$, we complete the system with the PDE for the substrate,
	\begin{equation}
		\label{eq:derivation_FinalC}
		\frac{\partial}{\partial t} c =  D  \frac{\partial^2}{\partial x^2} 
		+ \sum_{i=1}^n \frac{r_i}{Y_i} \frac{c}{K_i +c} .
	\end{equation}
	System \eqref{eq:derivation_FinalU}-\eqref{eq:derivation_FinalC} constitutes the one-dimensional version of the $n$-species competition model \eqref{eq:comp3}.
	
	\begin{remark} 
		Although this derivation is performed only for a subsystem of \eqref{eq:generalPDEmodel} in one spatial dimension, two-dimensional cases with different ecological interaction mechanisms can be derived similarly. For the derivation of the general PDE model \eqref{eq:generalPDEmodel}, one can consider an arbitrary number of substrates and introduce generic functions for the probabilities of bacterial growth and substrate consumption, together with an additional probability for substrate production. The subsequent steps of the derivation proceed analogously.
	\end{remark}
	
	\subsection{Local well-posedness} \label{sec:well_posedness}
	In this section, we provide a local-in-time well-posedness result for a particular case of the competition system \eqref{eq:comp3} for the total biomass. The main result follows from Amann's abstract framework for quasilinear parabolic boundary value problems, \cite{Am90}. 
	
	We consider the competition system \eqref{eq:comp3} with parameter symmetry, i.e., $d:= d_i=d_j$, $a:= a_i = a_j$, $r:= r_i=r_j$, $K:= K_i = K_j$, $Y:= Y_i = Y_j$, for all $i \neq j$, $1\leq i,j \leq n$, for the total biomass density $\rho(x,t):= \sum_{i}^nu_i (x,t)$, with $(x,t) \in \Omega \times [0, \infty)$, and where $\Omega \subset \R^2$ is an open set. A single substrate $c(x,t)$  is consumed by all the species. This system is governed by 
	\begin{align}
		\label{eq:comp_1_2d}
		\begin{split}
			\frac{\partial }{\partial t} \rho &= \nabla \cdot \left[ (d + a \rho) \nabla \rho \right] + f(\rho, c), \quad x \in \Omega, \, t>0,\\
			\frac{\partial }{\partial t} c &= D \Delta  c - g(\rho, c), \qquad \qquad \qquad x \in \Omega, \, t>0,
		\end{split}
	\end{align} where $d, a,r,K,b, Y$ are all positive constants and the reaction terms $f(\rho, c)$ and $g(\rho,c)$ are given by 
	\begin{align}\label{eq:reactions}
		f(\rho, c) = 	\rho \lp r \frac{c}{K + c} - b \rho \rp, \qquad g(\rho, c) = \rho \frac{r}{Y} \frac{c}{K + c}.
	\end{align}
	System \eqref{eq:comp_1_2d} is complemented with the initial and homogeneous Neumann boundary conditions,
	\begin{align} \label{eq:comp_1_2d_init_bdry}
		\begin{alignedat}{2}
			\rho (x, 0)&= \rho_0 (x) \geq 0, \quad x \in \Omega,  \qquad  \qquad \frac{\p}{\p \nu} \rho(x,t) &= 0, \quad  \text{  on } \partial \Omega\\
			c (x,0) &= c_0 (x) \geq 0, \quad  x \in \Omega,  \qquad  \qquad  \frac{\p}{\p \nu} c(x,t) &= 0,  \quad \text{  on } \partial \Omega,
		\end{alignedat}
	\end{align} where $\nu$ denotes the outward normal unit vector.
	We cast this system into the general framework followed in \cite{Am90}. Defining $u : = (\rho , c)^\top : \Omega  \to \R^2$, System \eqref{eq:comp_1_2d}-\eqref{eq:comp_1_2d_init_bdry} can be rewritten as
	\begin{align}	\label{eq:divergence_form}
		\begin{alignedat}{2}
			\frac{\p}{\p t} u  + \mc A(u) u &= F(u), \quad &x \in \Omega, \, t >0,\\ 
			\mc B(u)u &= 0, \quad &\text{on } \p \Omega, \, t\geq 0, \\
			u (x, 0) &= u_0(x),  \quad &x \in \Omega, \,  t =0, 
		\end{alignedat} 
	\end{align}where $(\mc A (u), \mc B (u))$ is of separated divergence form, i.e, defining the matrix $A(u):=\lp \begin{smallmatrix} d + a \rho & 0 \\ 0 & D\end{smallmatrix} \rp $, we can write 
	\begin{align*}
		\mc A(u) u = -\nabla \cdot \lp A(u) \nabla u\rp = 
		\begin{pmatrix}
			- \nabla \cdot \left [  \lp d + a \rho \rp \nabla \rho \right] \\
			- D \Delta c
		\end{pmatrix}, \qquad 
		\mc B (u) u = \lp A(u )\frac{\p}{\p \nu} u\rp =0,
	\end{align*} and $F(u) = \lp f(\rho, c), -g(\rho,c) \rp^\top $. We choose an open set $G := \left \{(\rho, c) \in \R^2 \, : \, \rho > -d/ 2a , \, c > -K/2\right \}$ which contains the biologically relevant region $(\rho, c) \in [0, \infty) \times [0, \infty) $. On this set, the eigenvalues of $A$, $\lambda_1 = d + a \rho \geq d/2>0$ and $\lambda_2 =D >0$ are both strictly positive and uniformly bounded away from $0$. Therefore, $(\mc A (u), \mc B (u))$ is uniformly elliptic and all three assumptions of \cite{Am90} are satisfied. 
\begin{remark}
	In particular, our system \eqref{eq:comp_1_2d}-\eqref{eq:comp_1_2d_init_bdry} fits within the general SKT (Shigesada-Kawasaki-Teramoto) cross-diffusion system given as a demonstrating example in \cite[Equation 7]{Am90} where $(v,w) = (\rho,c)$ and $\alpha_1 = d$, $\beta_{11} = a/2$, $\alpha_2 = D$, $ \beta_{12} = \beta_{21} = \beta_{22} = \gamma_{1} = \gamma_{2} = 0$, $h_1 (\rho, c) = f(\rho, c)/ \rho$, $ h_2(\rho,c) = -g(\rho,c)/ c$. 
\end{remark}
Finally, we define $V:=  \{ u \in H^{1,p}  \, :  \, u \in G \text{  for all  } x \in \Omega, \}$ and state the main result of this section: 
	
	\begin{thm}[Local-in-time well-posedness for \eqref{eq:comp_1_2d}-\eqref{eq:comp_1_2d_init_bdry}] 
		Let $\Omega \subset \R^2$ be a non-empty, bounded set with $C^2$ boundary. Consider System \eqref{eq:comp_1_2d}-\eqref{eq:comp_1_2d_init_bdry} where $d, a,r,K,b, Y$ are all positive constants, $(x,t) \in \Omega \times (0,\infty)$, with $u_0 = (\rho_0, c_0 )\in V$. Suppose that $\rho_0 \geq 0, c_0\geq 0$ and $p>2$. Then there exists a maximal time $T_{\text{max}} >0$ and a unique solution $u = (\rho,c) \in C ( [ 0, T_{\text{max}}); W^{1,p} ) \cap C^\infty (\bar \Omega \times (0, T_{\text{max}}); \mathbb{R}^2 ) $ of \eqref{eq:comp_1_2d}-\eqref{eq:comp_1_2d_init_bdry}. Moreover, $\rho (x,t) \geq 0, c(x,t) \geq 0$ for all $(x,t) \in \bar \Omega \times [0, T_{\text{max}})$.
	\end{thm}
	
	\begin{proof}
	The existence of of solutions up to a maximal time $T_{\text{max}}$ follows directly from the theory on parabolic quasilinear boundary value problems in  \cite[Theorem 1]{Am90}. System \eqref{eq:comp_1_2d}-\eqref{eq:comp_1_2d_init_bdry} is of `separated divergence form' \eqref{eq:divergence_form} and satisfies the hypotheses for Theorem 1. Therefore, Theorem 1 yields a unique maximal solution in $C([0, T_{\text{max}}); W^{1,p} ) \cap C^\infty (\bar{\Omega} \times (0, T_{\text{max}}); \mathbb{R}^2 ) $ of \eqref{eq:comp_1_2d}-\eqref{eq:comp_1_2d_init_bdry} where $0 < T_{\max} \leq \infty$.

    	To prove that solutions remain nonnegative over time, i.e, $\rho (x,t) \geq 0, c(x,t)\geq 0$ for all $(x,t) \in \bar \Omega \times (0, T_{\text{max}})$, we first define the negative parts of $\rho$ and $c$ by $\rho^-:=\max \{-\rho, 0\}$ and $c^-:=\max \{-c, 0\}$, respectively. We multiply the first equation in \eqref{eq:comp3} with $\rho^-$ and integrate over $\Omega$ (since $\rho$ is a classical solution for $t>0$), i.e.,
		\begin{align*}
			\int_{\Omega} \frac{\partial}{\partial t} \rho  \rho^{-}  \dx = \int_{\Omega} \nabla \cdot \lp \lp d + a \rho \rp \nabla \rho \rp \rho^{-} \dx + \int_{\Omega} \rho \lp \frac{r c}{K+c} - b \rho \rp  \rho^{-}\dx.
		\end{align*} 
		Since on the support of $\rho^-$, $\rho=-\rho^-$ we have 
		\begin{align*}
			\frac{1}{2} \frac{ \mathrm{d}}{ \mathrm{d} t } \int_{\Omega} \lp \rho^{-}\rp ^2 \dx = \int_{\Omega} \nabla \cdot \lp \lp d - a \rho^- \rp \nabla \rho^- \rp \rho^{-} \dx + \int_{\Omega} \lp \frac{r c}{K+c} + b \rho^- \rp  (\rho^{-})^2\dx.
		\end{align*}
	Since $\rho_0 \geq 0$, $d+ a \rho _0>0$ and $\rho$ is continuous, there exists a maximal time, $t_\rho >0$ such that for all $t \in (0, t_\rho)$ we have $\rho^- <d/(2a)$. Similarly, there exists a maximal time $t_c>0$ up to which $c^- < K/2$ holds. Now we take $\bar T: = \min \{t_\rho, t_c, T_{\text{max}} \}$, and argue that $\bar T = T_{\text{max}}$. Otherwise $\bar T < T_{\text{max}}$. For the first term on the right hand side, integration by parts and boundary conditions yield
		\begin{align*}
			\int_{\Omega} \nabla \cdot \lp \lp d - a \rho^- \rp \nabla \rho^- \rp \rho^{-} \dx = - \int_{\Omega} (d - a \rho^-) |\nabla \rho ^-|^2 \dx \leq 0,  \quad \text{ for any  } t \in [0, \bar T].
		\end{align*} Moreover, we also have 
		\begin{align*}
			\int_{\Omega} \lp \frac{r c}{K+c} + b \rho^- \rp  (\rho^{-})^2\dx \leq C_1 \int_{\Omega} (\rho^{-})^2\dx, \quad \text{ for any  } t \in [0, \bar T].
		\end{align*} where $C_1$ is a constant, since $\rho,c$ are bounded on $(x,t) \in \bar \Omega \times [0, \bar T]$. This gives 
		\begin{align*}
			\frac{ \mathrm{d}}{ \mathrm{d} t } \int_{\Omega} \lp \rho^{-}\rp ^2 \dx  \leq C \int_{\Omega} (\rho^{-})^2\dx, \quad \text{ for any  } t \in [0, \bar T].
		\end{align*}Gr\"onwall's lemma then yields
		\begin{align*}
			\int_{\Omega} (\rho^{-}(\cdot, t))^2\dx \leq \exp(C \bar T) \int_{\Omega} (\rho_0^{-})^2\dx  =0, \quad \text{for any  } t \in [0, \bar T]. 
		\end{align*} since $\rho_0^- = \rho^-(\cdot, 0)\equiv 0$. This implies that $\rho^- \equiv 0$ for any $t \in [0, \bar T]$; thus $t_\rho > \bar T$ since $t_\rho$ was defined as the maximal time such that $\rho^- <d/(2a)$. Using a similar argument we can show that $c^-\equiv 0$ for any $t \in [0, \bar T]$ and this implies $t_c > \bar T$. Therefore $\bar T = T_{\text{max}}$. 
	\end{proof}

   \begin{remark}
	Since the primary aim for this paper is to introduce our modelling framework, we restrict our analysis here to the local-in-time existence result. This level of detail suffices both for our numerical simulations and the real experiments conducted in bioreactors. However, to show that the solutions of  \eqref{eq:comp_1_2d}-\eqref{eq:comp_1_2d_init_bdry} exist globally in time, it is enough to prove that $ \sup _{0<t < T_{\max} } \| u(t) \|_{W^{1,p}} < \infty$, see \cite[page~18]{Am90}.  
	\end{remark}
	
	\subsection{Travelling wave analysis} \label{sec:travelling_wave}

	In this section, we carry out a formal travelling wave analysis in one spatial dimension of the coupled system for the total biomass and a substrate \eqref{eq:comp_1_2d}-\eqref{eq:reactions}, a parameter-symmetric variant of the competition case \eqref{eq:comp3}. The motivation for this analysis is the noticeable radially symmetric expansion of active biomass observed in Figures~\ref{fig:comp3} and \ref{fig:commensalism_u}. The computed minimal wave speed is later used for comparison with numerically measured propagation speed.
	
	The spatial expansion is driven by proliferation and subsequent taxis of biomass down the gradient of the population density. This movement process naturally creates two separate subdomains: the non-invaded region, where the substrate is close to its initial concentration and biomass is close to zero; and the invaded region where the biomass has started to deplete the substrate. Although the exact distribution of species within the total biomass depends on the initial conditions, the expansion effect for the total biomass remains the same.    
	
	Note that in one spatial dimension, i.e., $x \in \R$, system \eqref{eq:comp_1_2d} reads as
	\begin{equation}
		\label{eq:comp_wave}
		\begin{split}
			\frac{\partial }{\partial t} \rho &= \frac{\partial }{\partial x} \left[ (d + a \rho) \frac{\partial }{\partial x} \rho \right] + f(\rho, c), \\
			\frac{\partial }{\partial t} c &= D \frac{\partial ^2 }{\partial x^2 } c - g(\rho, c),
		\end{split}
	\end{equation} where $d, a,r,K,b, Y$ are all positive constants and the reaction terms $f(\rho, c)$ and $g(\rho,c)$ are given by \eqref{eq:reactions}. Moreover, \eqref{eq:comp_wave} is complemented with 
	\begin{equation} 	\label{eq:comp_wave_init_bdry}
		\begin{alignedat}{2}
			\rho (0,x) &= \begin{cases}
				\rho_0 (x)> 0, \, \, \text{for} \, x \leq 0, \\
				0 , \qquad \qquad  \text{for} \,  x > 0,
			\end{cases}   \qquad  &\frac{\p}{\p \nu} \rho(t,0) = 0, \,  \, \text{for} \,  t \geq 0\\
			c (0,x) &= c_0 (x)>0, \, \, \, \,  \,  \, \,\text{for} \, x \in \R,  \qquad  \qquad  \, \, \, &\frac{\p}{\p \nu} c(t,0) = 0,  \, \, \text{for} \, t \geq 0.
		\end{alignedat}
	\end{equation}
	Plugging a travelling wave ansatz $z:= x - vt, $ into \eqref{eq:comp_wave}-\eqref{eq:comp_wave_init_bdry} with traveling wave speed $v \in \R$, so that $U(x-vt) = \rho(x,t)$ and $C(x-vt) = c(x,t)$, where $U,C: \mathbb{R}\to \mathbb{R}$ are continuous and non-negative functions. We obtain the following system of ordinary differential equations, 
	\begin{equation}
		\label{eq:z_system}
		\begin{split}
			-vU'&= dU'' + a \left[ (U')^2 +  UU'' \right] + f(U,C),\\
			- v C' &= D C'' - g(U,C).
		\end{split}
	\end{equation}
	where the reaction terms are defined by \eqref{eq:reactions}. 
	
	When we use the ansatz $z= x-vt$, we implicitly lose the time dependence of the variable. However, we can identify asymptotic conditions ahead of the traveling wave in the limit as $z \to +\infty$,
	\begin{equation*}
		\begin{split}
			\lim_{z \to + \infty} \left( U(z), C(z) \right) = \left(0, c_0 \right), 
		\end{split}
	\end{equation*}
	where $c_0$ is the initial concentration of the substrate. This limit describes asymptotically the behaviour of the system ahead of the traveling wave - in the non-invaded region - where the substrate concentration is close to the asymptotic concentration $C \approx c_0$ and $U$ is small. We then consider the following asymptotic expansion about the (unstable) steady state $(0,c_0)$,
	\begin{align*}
		\begin{split}
			U &\approx \eps U_1 + \eps^2 U_2 + \cdots,\\
			C &\approx c_0 + \eps C_1 + \eps^2 C_2 + \cdots, 
		\end{split}
	\end{align*}
	where $\eps>0$ is a small constant. Substituting these into \eqref{eq:z_system}, we obtain at the first order of $\eps$,
	\begin{align*}
		\begin{split}
			d U_1'' + v U_1' + r U_1 \frac{c_0}{c_0 + K} &=0,\\
			D C_1'' + vC_1'-  U_1 \frac{r}{Y} \frac{c_0}{c_0 + K} &=0.
		\end{split}
	\end{align*}
	Notice that the first equation is now decoupled from the second equation and in particular, the first equation corresponds to the linearisation of the FKPP equation around the unstable steady state. We consider the ansatz $U_1 (z) = u_1^0 e^{ -\lambda z}$, $u_1^0>0, \lambda >0$ with $\lambda$ satisfying  
	\begin{equation*}
		\lambda^2 - \lambda \frac{v}{d} + \frac{r}{d} \frac{c_0}{c_0 + K} = 0,
	\end{equation*}
	which has the roots $\lambda_{1,2} = \frac{v}{2d} \pm \frac{1}{2d}\sqrt{v^2-4dr \frac{c_0}{c_0+K}}$. In case of complex $\lambda$, the solution oscillates around the steady state $0$, which is not biologically meaningful. Therefore, biologically meaningful solutions have a minimal traveling wave speed
	\begin{equation}
		\label{eq:wave_speed}
		v_{\text{min}} := 2  \sqrt{dr \frac{c_0}{K+c_0}}  .  
	\end{equation}
	
	\section{Numerical results} \label{sec:results}
	
	In this section, we show that our PDE framework \eqref{eq:generalPDEmodel}-\eqref{eq:generalPDE_IC_BC} for the special cases introduced in Section \ref{sec:cases} reproduces key behaviours that have previously been observed through IBMs \cite{K04, MASSG22}. We focus on three scenarios: competition in a three-species model without advantage of one over the others, competition between two species with different growth-yield strategies, and commensalism between three species (cf. Figure~\ref{fig:schematic}). In addition, we compare the travelling wave speed that we measure in numerical simulations for System \eqref{eq:comp_wave}-\eqref{eq:comp_wave_init_bdry} with the minimal travelling wave speed \eqref{eq:wave_speed}.

	\subsection{Competition}
	\label{ch:results_competition}
	We investigate competition scenarios from two perspectives: first, the dynamics of three competing populations within a microbial aggregate, and second, the interaction between two microbial populations with contrasting survival strategies. Both cases will be qualitatively compared to their individual-based counterpart in the existing literature
	(cf. \cite{MASSG22} for the three microaggregate species, \cite{K04} for different growth strategies).
	
	\paragraph{(A) Three-species competiton} \label{ch:results_competition3}

	In this competition setting, we consider three bacterial populations that rely on a single externally provided substrate, i.e. System \eqref{eq:comp3} for $n=3$. To ensure that no population has an intrinsic advantage, similarly to the setup in \cite{MASSG22}, we impose symmetry in the model parameters, with identical growth rates and half-saturation concentrations for all species ($r_1=r_2=r_3$, $K_1=K_2=K_3$; cf.~Table~\ref{tab:parameters} Case (A)). The initial condition (cf.~Figure~\ref{fig:comp3}) consists of three distinct inocula placed within a central circular region of the domain, while the substrate is distributed uniformly throughout the domain and supplied equally through all boundaries. The nature of this setup leads to the observed radial spread of all three populations.
	
		\begin{figure}[ht!]
		\centering
		\includegraphics[width=0.65\linewidth]{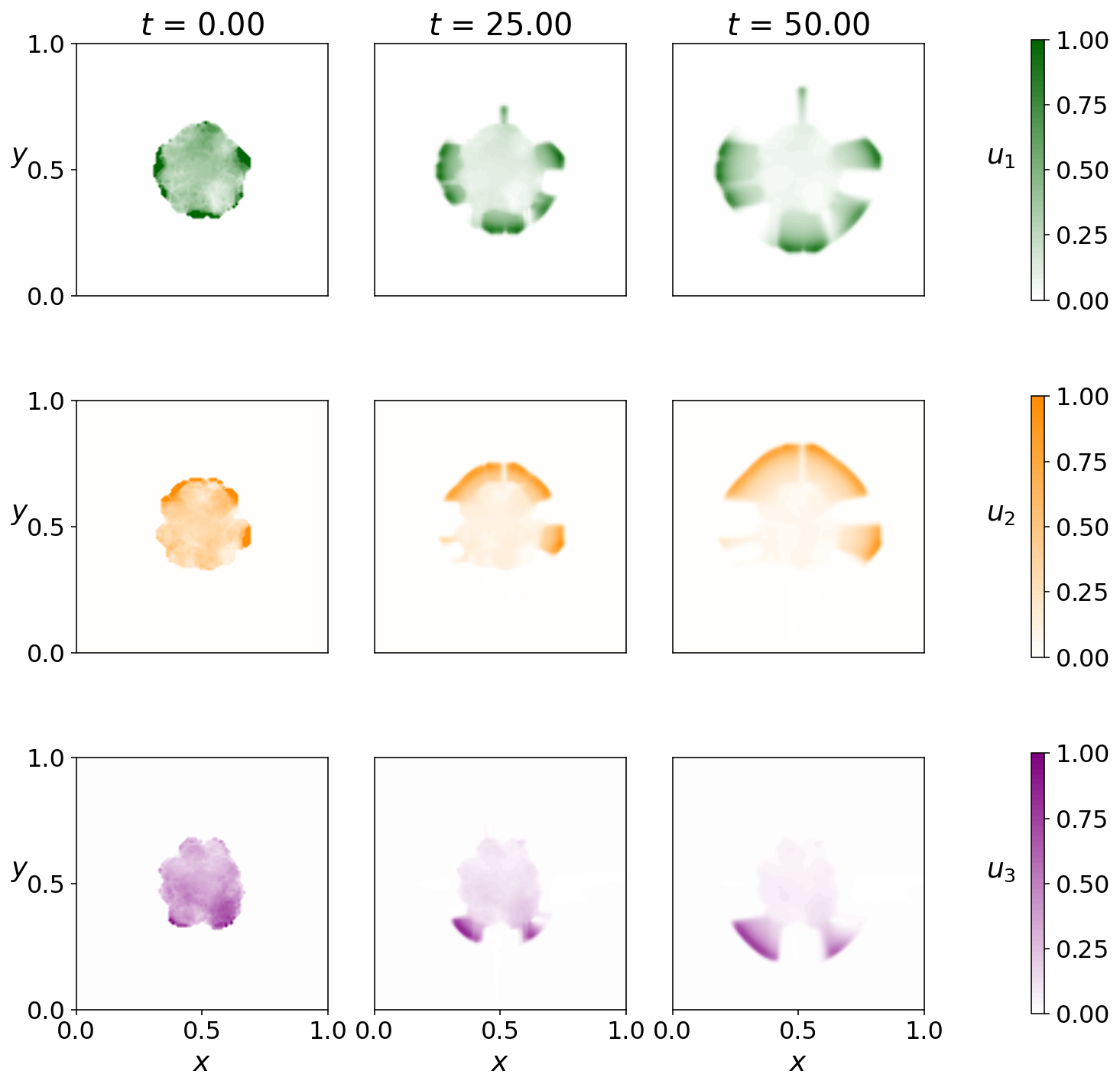}
		\caption{Case (A): Competition between three species \eqref{eq:comp3} with parameter values taken from Table~\ref{tab:parameters}. Population densities $u_1$ (first row, green), $u_2$ (second row, orange), and $u_3$ (third row, purple) at three distinct times $t \in \{0, 25, 50\}$.}
		\label{fig:comp3}
	\end{figure}

	The populations organise into sectorial patterns, with each species dominating in different angular sectors of the expanding total population. Within these sectors, higher biomass densities are observed at the advancing outer edge, while lower densities are found in the centre because of substrate depletion and subsequent bacterial decay. The main driver for this phenomenon appears to be the localised growth laws, i.e., higher density areas are reinforced by their own growth over time. 	In\cite{MASSG22}, the authors observe interspecific segregation of communities in the IBM that they developed. They call this phenomenon `columned stratification': the distribution of the species into radially growing sectors. This is precisely what we observe with the PDE model as well. To enable a direct comparison with the results reported in \cite{MASSG22}, we overlaid population distributions exceeding a prescribed density threshold in Figure \ref{fig:competition_overlap} which can be compared to Figure 2, `Competition' column and Figure 4C in \cite{MASSG22}.

	\begin{figure}[ht!]
		\centering
		\begin{subfigure}[b]{0.45\linewidth}
			\centering
			\includegraphics[width=\linewidth]{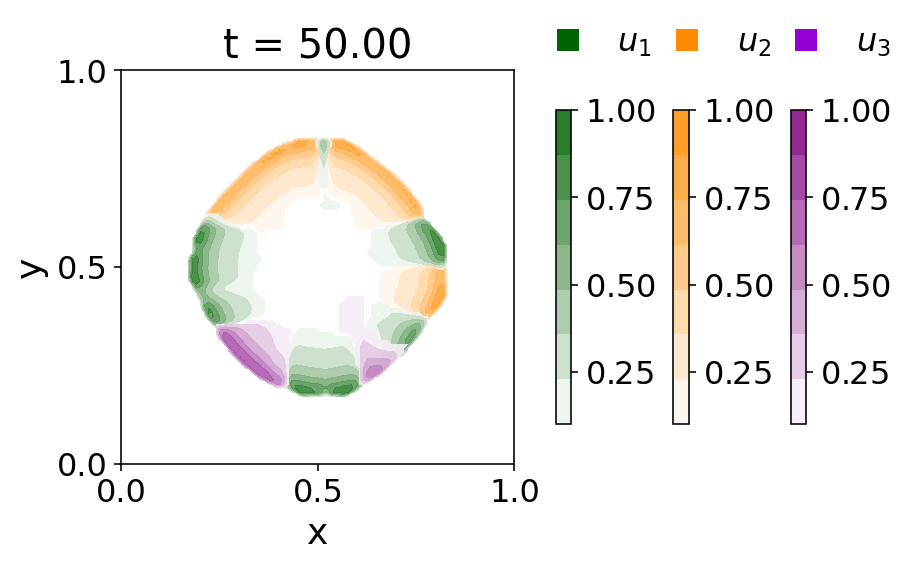}
			\caption{Competition scenario \textbf{(A)}}
			\label{fig:competition_overlap}
		\end{subfigure}
		\hfill
		\begin{subfigure}[b]{0.45\linewidth}
			\centering
			\includegraphics[width=\linewidth]{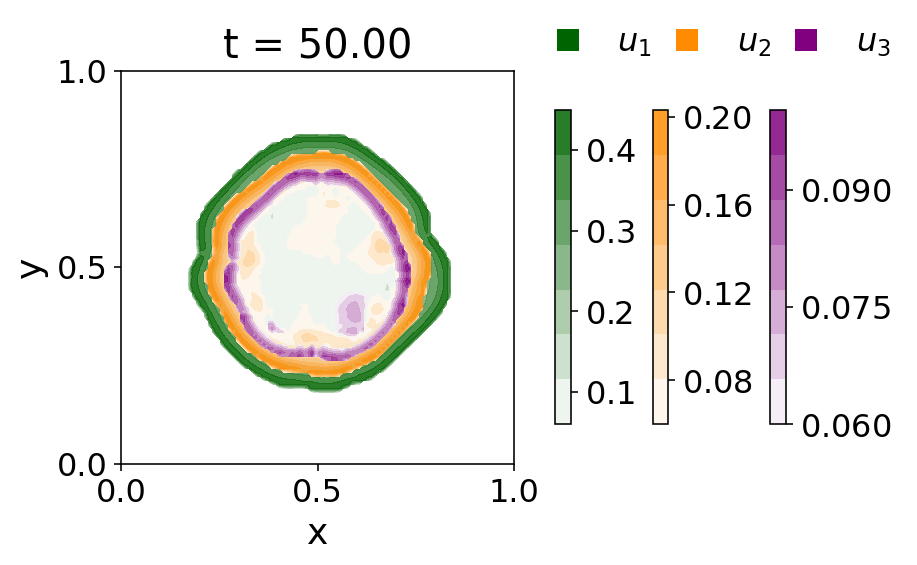}
			\caption{Commensalism scenario \textbf{(C)}}
			\label{fig:commensalism_overlap}
		\end{subfigure}
		\caption{Overlap of three population densitites at the final time $t=50$ under two scenarios: (a) Case (A): Competition between three species; (b) Case (C): commensalism of three species; Densities below $0.1$ (for (A)) and $0.07$ (for (C)) were omitted to enhance the visibility of the distinct colours. Contours are shown at several discrete levels, with darker shades indicating higher densities.}
		\label{fig:overlap_comparison}
	\end{figure}
	
	\paragraph{(B) Trade-off scenario: Growth strategist vs. yield strategist} \label{ch:results_growth_vs_yield}
	In this case, we consider the setting with two competing species with opposing growth strategies, described by the system~\eqref{eq:comp3} for $n=2$. Symmetry is imposed on the parameters, i.e., $d_1=d_2$, $a_1=a_2$, $K_1=K_2$ and $b_1=b_2=0$, while asymmetry is introduced in the growth rate and yield coefficients, with $r_1 < r_2$ and $Y_1 > Y_2$. Hence, population $u_1$ represents a yield strategist (YS), characterised by a more economical use of substrate, whereas population $u_2$ represents a growth strategist (GS), defined by a higher maximum growth rate.

		\begin{figure}[ht!]
		\centering
		\begin{minipage}{\textwidth}
			\begin{subfigure}[t]{0.49\textwidth}
				\centering
				\includegraphics[width=\linewidth]{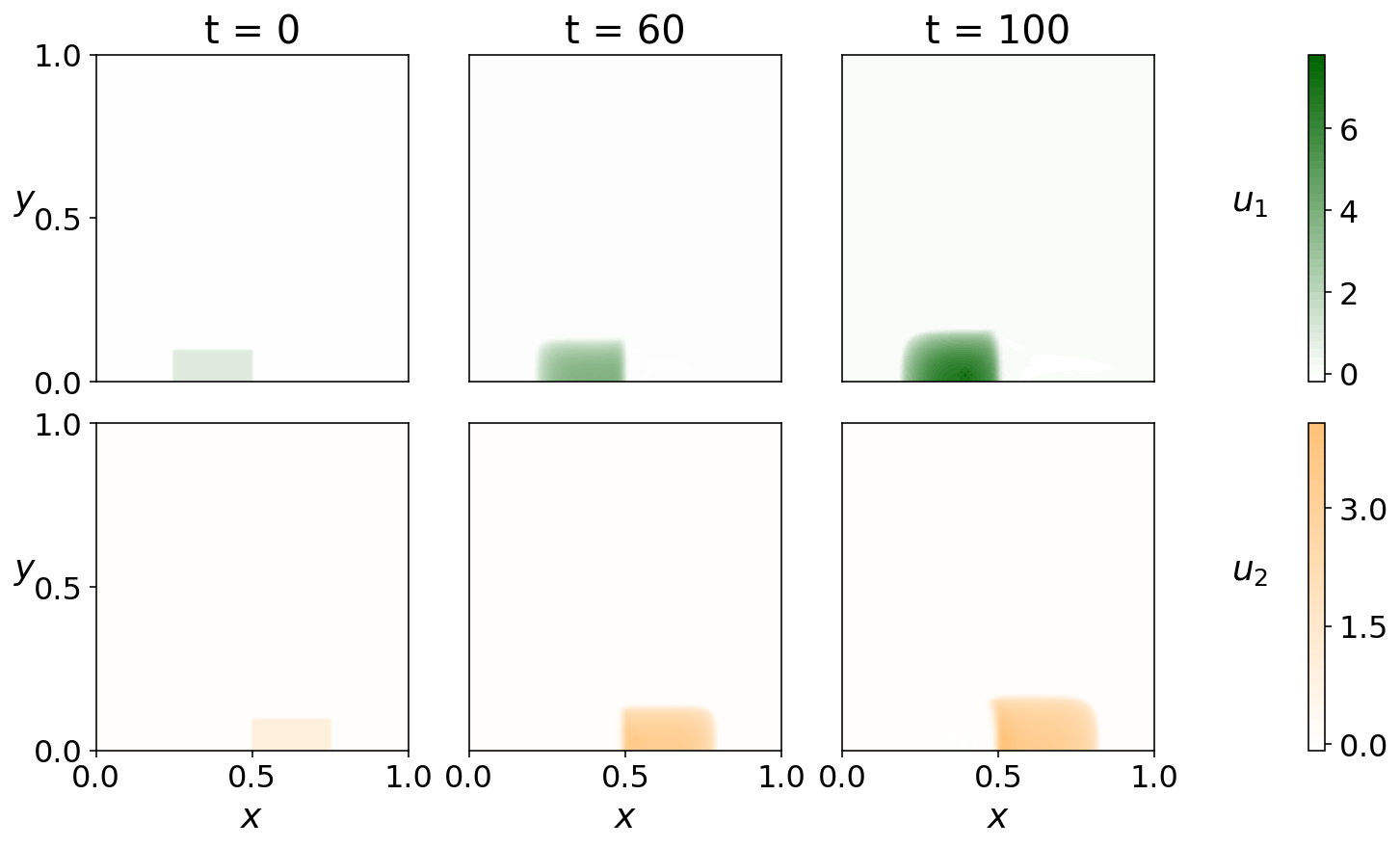} 
				\caption{Initial setup in two blocks}
				\label{fig:comp_kreft_spatial_blocks}
			\end{subfigure}
			\hfill
			\begin{subfigure}[t]{0.49\textwidth}
				\centering
				\includegraphics[width=\linewidth]{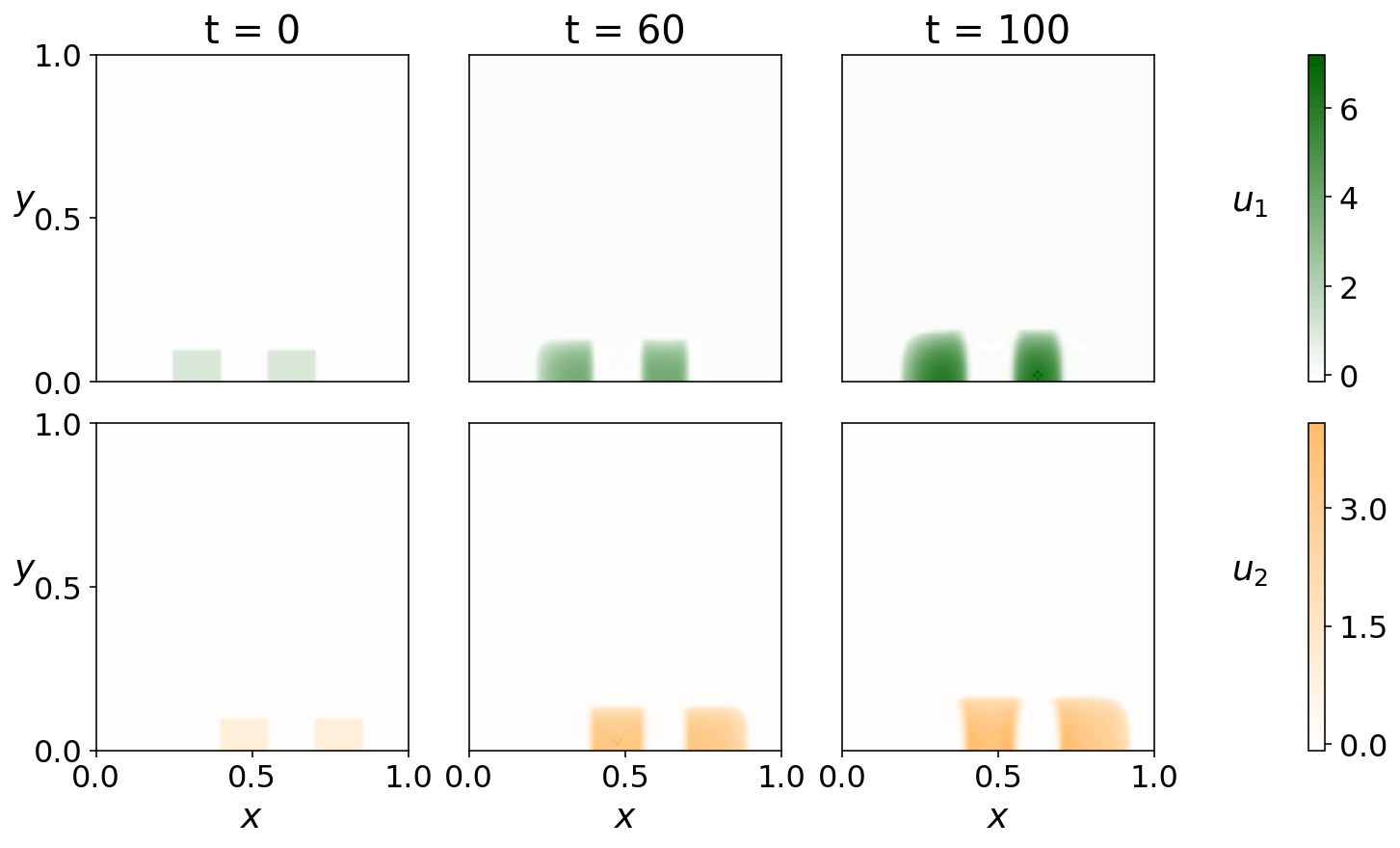} 
				\caption{Initial setup in four stripes}
				\label{fig:comp_kreft_spatial_stripes}
			\end{subfigure}
		\end{minipage}
          \par\vspace{1.5em}
		\begin{minipage}{\textwidth}
			\begin{subfigure}[t]{0.4\textwidth}
				\centering
				\includegraphics[width=0.8\linewidth]{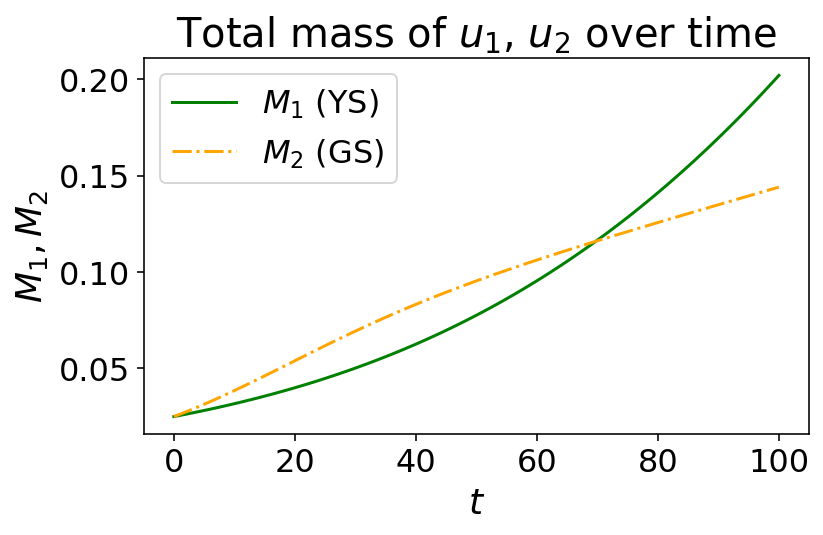} 
				\caption{Total mass for (a) }
				\label{fig:comp_kreft_lumped_blocks}
			\end{subfigure}
			\qquad \qquad 
			\begin{subfigure}[t]{0.4\textwidth}
				\centering
				\includegraphics[width=0.8\linewidth]{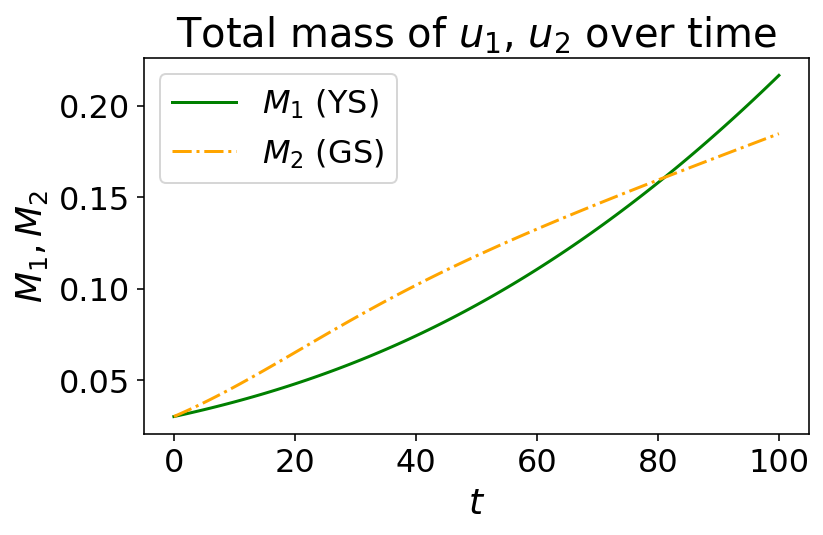}
				\caption{Total mass for (b)}
				\label{fig:comp_kreft_lumped_stripes}
			\end{subfigure}
		\end{minipage}
	 \caption{Case (B): Yield strategist (YS) vs. growth strategist (GS) with parameter values taken from Table~\ref{tab:parameters}. Figures (a), (b): Population densities $u_1$ (green), $u_2$ (orange) at times $t \in \{ 0, 60, 100\}$. Figures (c), (d): Total masses $M_1 := \int _\Omega u_1 (x,t) \dx$ (YS, green), $M_2 := \int _\Omega u_2 (x,t) \dx$ (GS, orange dashed) over time for two scenarios (a), (b) respectively.}
  \label{fig:growth_vs_yield}
		\end{figure}

	To replicate the setup of~\cite{K04}, the system is initialised with adjacent strips of $u_1$ and $u_2$ along the bottom boundary, and the substrate is provided exclusively through the top boundary. We perform one simulation with two blocks of initial populations, and one with a total of four alternating stripes of populations (cf. Figures~\ref{fig:comp_kreft_spatial_blocks},~\ref{fig:comp_kreft_spatial_stripes}). The simulations show that the yield strategist gains a competitive advantage over time (cf.~Figures~\ref{fig:comp_kreft_lumped_blocks},~\ref{fig:comp_kreft_lumped_stripes}). This advantage arises from its higher substrate efficiency: $u_1$ requires fewer nutrients for growth per unit biomass, allowing it to sustain larger populations in regions of limited substrate availability. 
	This outcome is, however, sensitive to the substrate diffusion coefficient. For sufficiently large substrate diffusion, the advantage of the yield strategist diminishes, as nutrients are rapidly redistributed across the domain and the growth strategist is able to compete more effectively. 
	Even though the GS initially seems to outgrow the YS, after some critical time the nutrient has been depleted enough around the GS and the YS starts to exceed the GS in terms of total biomass due to its frugal nutrient usage (cf. Figure~\ref{fig:growth_vs_yield}). Although more detailed initial setups are considered in \cite{K04}, qualitatively our results can be compared to those in e.g., Figure 4 in \cite{K04}.

	\begingroup
	\renewcommand{\thesubsection}{(C)}
	\subsection{Commensalism} \label{ch:results_commensalism}
	\endgroup
	
	\noindent For this case, we consider the commensal relationship of three bacterial populations via three substrates, i.e., System \eqref{eq:comm} (cf. Figure \ref{fig:schematic_commensalism}). Similar to the case (A), the initial conditions for the bacterial species are set up in a central circular region of the domain (cf. Figure~\ref{fig:commensalism_u}). From the substrates, only the first one is initially supplied, while substrates $c_2$ and $c_3$ are created exclusively as metabolic by-products of populations $u_1$ and $u_2$, respectively. The setting was posed in this way to allow qualitative comparison to the commensalism scenario in the IBM \cite{MASSG22}.
	
	\begin{figure}[ht!]
		\centering
		\includegraphics[width=0.7\linewidth]{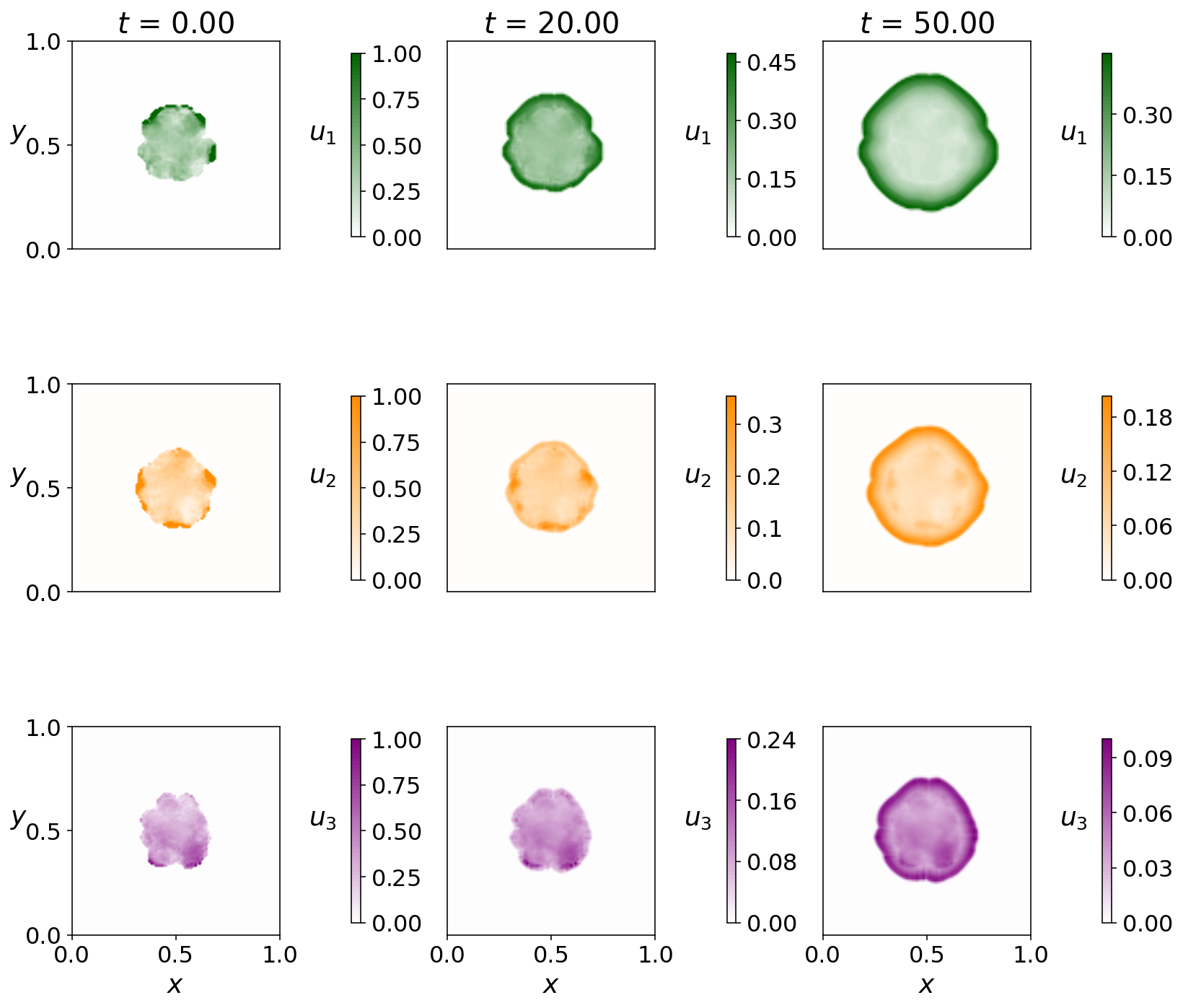}
		\caption{Case (C): Commensalism between three species \eqref{eq:comm} with parameter values taken from Table~\ref{tab:parameters}. Population densities $u_1$ (first row, green), $u_2$ (second row, orange), and $u_3$ (third row, purple) at three distinct times $t \in \{0, 25, 50\}$.}
		\label{fig:commensalism_u}
	\end{figure}
	Qualitatively, our results exhibit the same `layered stratification' effect as reported in \cite{MASSG22}. For a direct comparison we again overlay the populations above a chosen density threshold in Figure~\ref{fig:commensalism_overlap}, corresponding to the results displayed in Figure 5C in \cite{MASSG22}.
	
	The simulations reveal a radially spreading organisation of the populations into concentric rings. The outermost ring is formed by Species $1$, which consumes the supplied Substrate $1$. Immediately inside this layer, Species $2$ establishes itself, sustained by the metabolic product, Substrate $2$, of Species $1$. The innermost region is occupied by the Species $3$, which depends on the by-product, Substrate $3$ of Species $2$. This spatial ordering reflects the substrate consumption hierarchy inherent in the system: the species dependent on the externally supplied substrate dominates the periphery, successively engulfing the species whose growth relies on secondary metabolites. In all three species, the innermost populations have mostly decay due to the absence of substrates.

        \bigskip

	\subsection{Comparison with the minimum travelling wave speed} \label{sec:wave_speed_comparison}

	In this section, we perform numerical simulations of the parameter-symmetric competition system for the total biomass \eqref{eq:comp_wave}-\eqref{eq:reactions} for $x \in [0,200]$, $t \in [0,200]$ with homogeneous Neumann boundary conditions. We choose initial conditions $\rho_0 (x)= \exp(-x)$ and $ c_0 = 5$ to mimic a sharp decline of population and uniform substrate supply, respectively. All remaining simulation parameters are taken from Table~\ref{tab:parameters}, Case (A). We can indeed observe the emergence of traveling pulses into the positive $x$-direction (cf. Figure~\ref{fig:travwave}). The wave profile closely resembles the outward-growing ring of active biomass in the two-dimemsional simulations, see Figures \ref{fig:comp3} and \ref{fig:competition_overlap}. 
	Behind the 1D pulse - corresponding to the centre of the 2D population - the biomass decays due to substrate depletion. Only the outer ring of biomass remains active and continuously advances into the non-invaded substrate-rich region. For the wave speed measurement, the location where the total biomass attains its maximum $\bar{x} (t):= \text{argmax}_x \rho(x,t)$ is tracked over time. The measured value given here is computed from the final two timesteps displayed in Figure~\ref{fig:travwave},  i.e., $\bar{v} = \frac{\bar{x}(t=200) - \bar{x}(t=160)}{200-160} \approx  0.8396$. Using the system parameters, we also compute the minimal wave speed $v_{\text{min}} = 0.0018$ by the formula \eqref{eq:wave_speed}. Although these two values differ quantitatively, the measured speed is indeed, as expected, higher than the theoretical minimal wave speed. This indicates that biomass growth dominates the propagation and causes the wave to travel faster than the speed predicted through the linearisation.
	
	\begin{figure}[ht!]
		\centering
		\includegraphics[width=0.7\linewidth]{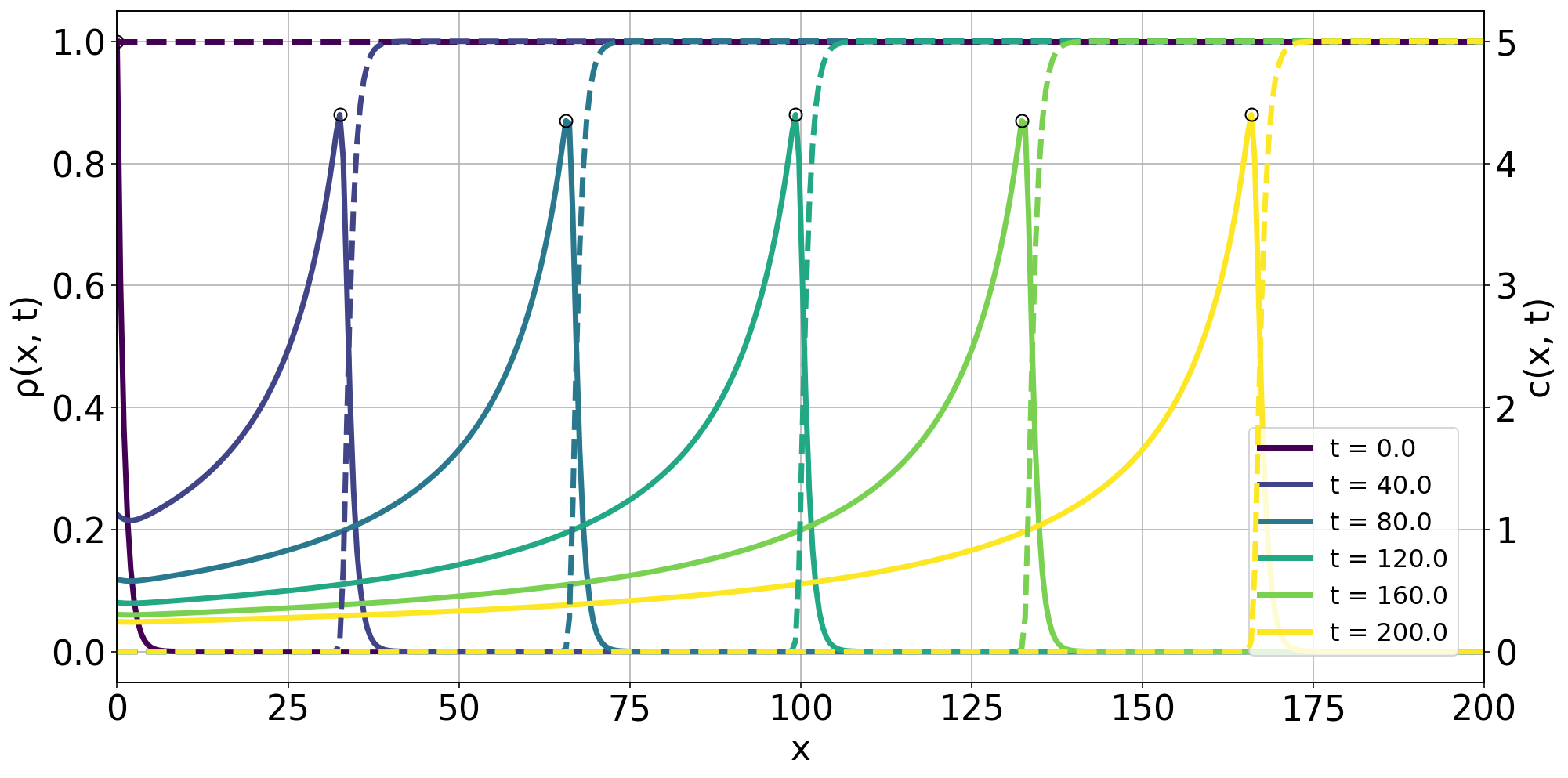}
		\caption{Travelling wave behaviour of the total population $\rho(x,t)$ (full lines) and the substrate $c(x,t)$ (dashed lines) plotted over space $x$ at different time points $t \in \{0,40,80,120,160,200\}$; The maximal value of the biomass is tracked (black circle) to measure the wave speed; Measured wave speed $\bar{v} \approx 0.8396$; Calculated wave speed $v_{\text{min}} = 0.0018$.}
		\label{fig:travwave}
	\end{figure}
	
	\section{Discussion and future perspectives} \label{sec:discussion}
	
	In this study, we derived a system of PDEs to describe the spatio-temporal dynamics of multiple bacterial species within a microbial aggregate. This model aimed to capture both their spatial organisation as well as inter-species ecological interactions through the consumption and/or production of substrates. Although the IBMs developed for the same dynamics, e.g. \cite{MASSG22, K04} have the power to offer detailed insight into the behaviour of individual bacteria, they are often computationally highly expensive. This limits their applicability for systematic exploration of large parameter spaces. Moreover, they offer very little room for mathematical analysis resulting in lack of predictive power for quantitative behaviour of the system. By formulating an associated PDE model, we were able to improve computational efficiency - from a few hours for the IBM to a few seconds for the PDE results - while still maintaining biological realism. We qualitatively validated the PDE model by comparing its results with a range of behaviours reported in earlier IBMs. We successfully reproduced several key features of microbial aggregates and biofilms that have previously been predicted in IBMs. 
	
	In competitive scenarios, we confirmed that yield strategists can outcompete growth strategists under nutrient-limited conditions. However, this advantage was highly sensitive to factors such as the initial composition of the aggregate and the diffusion coefficient of the substrate. Given that initial spatial configurations are difficult to control experimentally, particularly when it comes to microbial aggregates that form spontaneously, these findings suggest that the dominance of yield strategists in laboratory settings may be less robust than the theory alone would predict.
	
	When populations are metabolically similar, our simulations reproduced the emergence of spatial sector formation, where species segregate into distinct spatial domains. This phenomenon is largely explained by the local proliferation of bacteria, combined with passive displacement due to crowding and shoving that arises from cell proliferation. In commensal systems, our PDE model also captured the formation of concentric population rings, consistent with a metabolic hierarchy in which species align spatially according to their position within the consumption chain. All of our results closely resemble those obtained in previous individual-based studies for biofilms \cite{K04} and microbial aggreates \cite{MASSG22}, supporting the validity and relevance of our continuum model formulation. 
	
	Beyond reproducing qualitative ecological outcomes, the PDE formulation offers opportunities for mathematical analysis. For example, in two of our cases we observed radial expansion of aggregates from central inocula, resembling travelling wave phenomena. In a simplified setting, where a homogeneous biomass consumes a single nutrient, we were able to compute a minimal wave speed. Extending these results to more intricate scenarios involving multiple species and substrates with different metabolic strategies or ecological interactions, offers a promising direction for future work.

	The scenarios we investigated in this study were chosen to showcase the qualitative capabilities of our framework. For more realistic settings, however, additional modifications may be required. For instance, we impose Neumann boundary conditions for both, the bacterial species as well as the substrates. While this is a common modelling assumption, different boundary conditions such as time-dependent Dirichlet or Robin conditions may mimic laboratory settings more closely, since the bulk concentration can potentially be controlled. 
	Furthermore, the substrate diffusion occurs on a much faster timescale than substrate uptake and bacterial proliferation; and these processes themselves are again faster than cell motility. Therefore, it may be appropriate to separate different timescales and apply quasi-steady state assumptions for the substrate concentrations. Future work may involve the extension of our framework to more realistic settings, such as nitrifying communities, and conducting multi-scale or travelling wave analysis of more detailed multi-species systems. Last but not least, a comprehensive mathematical analysis of the general PDE system \eqref{eq:generalPDEmodel}-\eqref{eq:generalPDE_IC_BC} is also left for future work. 
	
	\section*{Acknowledgements}
	The authors would like to thank Eloi Martinez-Rabert for fruitful discussions and for providing references on the metabolic capabilities of nitrifying bacteria.
	
	\smallskip
	\noindent V. Freingruber is supported by the Delft Technology Fellowship of H. Yolda\c{s}.  H. Yolda\c{s} is supported by the Dutch Research Council (NWO) under the NWO-Talent Programme Veni ENW project MetaMathBio with the project number VI.Veni.222.288. Some part of this research was conducted while H. Yolda\c{s} was visiting the Okinawa Institute of Science and Technology (OIST) through the Theoretical Sciences Visiting Program (TSVP). Moreover, this research is supported by the Climate Action Research and Education Seed Fund (call November 2023) of the Climate Action Programme of TU Delft. For the purpose of open access, the authors have applied a Creative Commons Attribution (CC-BY) licence to any Author Accepted Manuscript version arising from this submission. Figures 1 and 2 were created with BioRender.com

	\section*{Code availability statement}
	The code used to generate the results in this article is publicly available at \newline 
	\href{https://github.com/vevaF/Microecological-Interactions}{https://github.com/vevaF/Microecological-Interactions}. 
	
	\appendix
	
	\section{Numerical methods} \label{app:numericalscheme}
	For the numerical solution of the two-dimensional problems the method of lines is used, i.e. the space is discretised and the resulting system of ODEs is solved using a pre-implemented explicit Runga-Kutta method of order 5(4) from the python package \textit{scipy.integrate.solve\_ivp}. 
	For the space discretization we use commonly finite volume schemes, see for example  \cite{HV03}. In particular, we use a central differences for the second order term and a first order upwinding scheme for the advection along the pressure gradient, with details given below.

	The space domain $[0,L]\times [0,L]$ was discretised in the following way: let $N+1$ denote the number of discretisation points in both the $x$- and $y$-direction (including both endpoints). Then the full domain is divided into $N^2$ squares with side lengths $\Delta x = \Delta y= \frac{L}{N}$.  We use $L=1$, $T= \{50, 100, 50\}$ (for (A), (B), (C) respectively), for the simulation window $(x,y,t)\in [0,L]\times[0,L]\times[0,T]$ and $N+1 = 100$ discretisation points in both spatial dimensions.
	
	Let $u^{i,j}_k(t) := u_k (x_i, x_j t) = u_k (i \Delta x, j \Delta y, t)$ denote the cell number density of the $k$-th bacterial species at position $(x_i,y_j)$, $1 \leq k \leq n$, and $c^{i,j}_h (t):= c_h (x_i, y_j, t) = c_h (i \Delta x, j \Delta y, t)$ the concentration of the $h$-th substrate at position $(x_i,y_j)$, $1 \leq h \leq m$, where $n$ and $m$ are the total numbers of bacterial species and substrates, respectively, and $0 \leq i,j, \leq N$. The density of the total biomass at location $(x_i,y_j)$ is defined as $\rho ^{i,j} (t) := \sum_{k=1}^{n} u_k^{i,j}(t)$.
	
	Let us denote the population flux at time $t$ from location $(x_i, y_j)$ to $(x_{i+1},y_j)$ as $f^{u_k}_{i+1/2, j}(t)$, for $0\leq i \leq N-1$, $0 \leq j \leq N$. Similarly, let $f^{u_k}_{i, j+1/2}(t)$ be the population flux at time $t$ from $(x_i, y_j)$ to $(x_{i},y_{j+1})$, $0 \leq i \leq N$, $0 \leq j \leq N-1$. We have
	\begin{align*}
		f_{i+1/2,j}^{u_k} (t) = \begin{cases}  - \frac{1}{\Delta x } d_k (u^{i+1,j}_k (t) - u^{i,j}_k (t) ) - \frac{1}{\Delta x } a_k (\rho^{i+1} (t) - \rho ^ i(t) ) u^{i,j}_k (t),
			\, &\text{if } \,\, \rho^{i+1}(t) - \rho ^i(t) >0, \\
			- \frac{1}{\Delta x } d_k (u^{i+1,j}_k(t) - u^{i,j}_k(t)) + \frac{1}{\Delta x } a_k (\rho^{i+1}(t) - \rho^i (t) ) u^{i+1,j}_k (t),
			\,  &\text{if } \,\, \rho ^{i+1}(t) - \rho^i (t) < 0,
		\end{cases}
	\end{align*}
	for $1 \leq i \leq N-1$, $ 1\leq j \leq N$, and
	\begin{align*}
		f_{i,j+1/2}^{u_k} (t) = \begin{cases}  - \frac{1}{\Delta y } d_k (u^{i,j+1}_k (t) - u^{i,j}_k (t) ) - \frac{1}{\Delta y } a_k (\rho^{i,j+1} (t) - \rho ^{i,j} (t) ) u^{i,j}_k (t),
			\, &\text{if} \, \,  \rho^{i+1,j}(t) - \rho^{i,j}(t) >0, \\
			- \frac{1}{\Delta y} d_k (u^{i,j+1}_k(t) - u^{i,j}_k(t)) + \frac{1}{\Delta y} a_k (\rho^{i,j+1}(t) - \rho^{i,j} (t) ) u^{i,j+1}_k (t),
			\, &\text{if} \, \,  \rho ^{i,j+1}(t) - \rho^{i,j} (t) < 0,
		\end{cases}
	\end{align*} for $1 \leq i \leq N$, $1\leq j \leq N-1$, where $d_k, a_k$ are the diffusion and advection coefficients of the $k$-th population, respectively. 
	
	We define $f^{u_k}_{-1/2,j} (t) = f^{u_k}_{N+1/2,j} (t) = f^{u_k}_{i,-1/2} (t) = f^{u_k}_{i, N+1/2} (t) = 0 $ for $1\leq k \leq n$, $0 \leq t \leq T$ to enforce homogeneous Neumann boundary conditions.
	
	Similarly, we define the flux of the $h$-th substrate $c_h$ from $(x_i,y_j)$ to $(x_{i+1},y_j)$ at time $t$ as
	
	\begin{equation*}
		f^{c_h}_{i+1/2,j} = -  D_h \frac{1}{\Delta x } (c_h^{i+1,j}(t)- c_h^{i,j}(t) ), 
	\end{equation*}
	for $1 \leq i \leq N-1$, $1 \leq j \leq N$, and from $(x_i,y_j)$ to $(x_{i},y_{j+1})$ at time $t$ as
	\begin{equation*}
		f^{c_h}_{i,j+1/2} = -  D_h \frac{1}{\Delta y} (c_h^{i,j+1}(t)- c_h^{i,j}(t) ), 
	\end{equation*}
	for $1 \leq i \leq N$, $1 \leq j \leq N-1$, where $D_h$ is the diffusion coefficient of the $h$-th substrate, with boundary conditions $f^{c_h}_{-1/2,j} (t) = f^{c_h}_{N+1/2,j} (t) = f^{c_h}_{i,-1/2} (t) =f^{c_h}_{i,N+1/2} (t) = c_h^\infty \geq 0$ for $1 \leq h \leq m$, $0 \leq t \leq T$. Note that these boundary conditions can also be easily adapted to have homogeneous Neumann conditions on specific boundaries and non-homogeneous conditions on others.
	
	\noindent The corresponding ODE system then reads as
	\begin{align*}
		\frac{\mathrm d}{ \mathrm d t} u_k ^{i,j}(t) &= \frac{1}{\Delta x } (f^{u_k}_{i-1/2,j}(t) - f^{u_k}_{i+1/2,j}(t)) + \frac{1}{\Delta y } (f^{u_k}_ {i,j-1/2} (t) - f^{u_k}_{i,j+1/2} (t))  + k^{u_k}_{i,j} (t, c_1^{i,j}, \cdots, c_m^{i,j}, u_1^{i,j}, \cdots u_n^{i,j}), \\
		\frac{\mathrm d }{\mathrm d t} c_h^{i,j}(t) &= \frac{1}{\Delta x } (f^{c_h} _{i-1/2,j}(t) - f^{c_h}_{i+1/2,j}(t)) + \frac{1}{\Delta y } (f^{c_h} _{i,j-1/2}(t) - f^{c_h}_{i,j+1/2}(t))+ k^{c_h}_{i,j}(t, c_1^{i,j}, \cdots, c_m^{i,j}, u_1^{i,j}, \cdots u_n^{i,j}),
	\end{align*} for $0 \leq i \leq N$, $0 \leq j \leq M$ and with appropriate initial conditions $u_k^{i,j}(0) \geq 0 $, $c_h^{i,j}(0) \geq 0$. Here, $k^{u_k}_{i,j}(t, c_1^{i,j}, \dots, c_m^{i,j}, u_1^{i,j}, \dots u_n^{i,j})$ and $k^{c_h}_{i,j}(t, c_1^{i,j}, \dots, c_m^{i,j}, u_1^{i,j}, \dots u_n^{i,j})$ are the reaction/ growth/ consumption/ production terms specific to the $k$-th species and $h$-th substrate, respectively. 

	\begin{table}[h!]
		\centering
		\begin{tabular}{lccc}
			\hline
			\textbf{Parameter} & \textbf{Case (A)} & \textbf{Case (B)} & \textbf{Case (C)} \\
			\hline
			$d_{1} = d_2 (= d_3) $ & $10^{-6}$ & $10^{-6}$ & $10^{-6}$ \\
			$a_{1} = a_2 (=a_3)$ & $10^{-5}$ & $10^{-5}$ &  $10^{-5}$\\
			$r_{1}$ & $1$ & 0.5 & 1 \\
			$r_{2}$ & $1$ & 1 &  1\\
			$r_{3}$ & $1$ & n.a. & 1 \\
			$K_{1}$ & $1$ &  100 &  1\\
			$K_{2}$ & $1$ &  100 &  1\\
			$K_{3}$ & $1$ & n.a. &  1\\
			$b_{1} = b_2 (=b_3) $ & $0.1$ & 0 & 0.1 \\
			$D$ & $10^{-4}$ & $10^{-5}$ & n.a.\\
			$D_1=D_2=D_3$ & n.a. & n.a. & $10^{-5}$\\
			$Y_{1}$ & $0.2$ & 4 & 10 \\
			$Y_{2}$ & $0.2$ & 0.5 &  10 \\
			$Y_{3}$ & $0.2$ & n.a. & 10 \\
			$\sigma_{1}$ & n.a. & n.a. &  0.5\\
			$\sigma_{2}$ & n.a. & n.a. &  0.5\\
			$L$ & $1$ & 1 & 1 \\
			$c^0$& 5 & 5 & n.a.\\
			$c^\infty$& $10^{-5}$  & 1& n.a. \\
			$c_1^0$& n.a. & n.a. & 5\\
			$c_1^\infty$& n.a. & n.a. &0 \\
			
			\hline
		\end{tabular}
		\caption{Parameter values (non-dimensional) used for the three simulation experiments inside a domain $\Omega = [0,L] \times [0,L]$: \textbf{(A)} Competition between three species;  \textbf{(B)} Competition between two species with different growth strategies;  \textbf{(C)} Commensalism between three species; (n.a. = not applicable; parameter not needed for this setting).}
		\label{tab:parameters}
	\end{table}

    \bibliography{EcologicalInteractions}

@article{AM89,
	ISSN = {00301299, 16000706},
	author = {Wallace Arthur and Paul Mitchell},
	journal = {Oikos},
	number = {1},
	pages = {141--143},
	publisher = {[Nordic Society Oikos, Wiley]},
	title = {A Revised Scheme for the Classification of Population Interactions},
	urldate = {2025-10-19},
	volume = {56},
	year = {1989}
}

@article{B94,
	ISSN = {00335770, 15397718},
	author = {Judith L. Bronstein},
	journal = {The Quarterly Review of Biology},
	number = {1},
	pages = {31--51},
	publisher = {The University of Chicago Press},
	title = {Our Current Understanding of Mutualism},
	urldate = {2025-10-19},
	volume = {69},
	year = {1994}
}

@article {BW85,
	AUTHOR = {Butler, Geoffrey J. and Wolkowicz, Gail S. K.},
	TITLE = {A mathematical model of the chemostat with a general class of
	functions describing nutrient uptake},
	JOURNAL = {SIAM Journal on Applied Mathematics},
	VOLUME = {45},
	YEAR = {1985},
	NUMBER = {1},
	PAGES = {138--151},
	ISSN = {0036-1399},
	MRCLASS = {92A08 (92A09)},
	MRNUMBER = {775486},
	MRREVIEWER = {A.\ Hausrath},
	DOI = {10.1137/0145006},
	URL = {https://doi.org/10.1137/0145006},
}

@article{C20,
	title={Non-surface attached bacterial aggregates: a ubiquitous third lifestyle},
	author={Cai, Yu-Ming},
	journal={Frontiers in Microbiology},
	volume={11},
	pages={557035},
	year={2020},
	publisher={Frontiers Media SA}
}

@article {CLM20,
	AUTHOR = {Chaplain, Mark A. J. and Lorenzi, Tommaso and Macfarlane,
	Fiona R.},
	TITLE = {Bridging the gap between individual-based and continuum models
	of growing cell populations},
	JOURNAL = {Journal of Mathematical Biology},
	VOLUME = {80},
	YEAR = {2020},
	NUMBER = {1-2},
	PAGES = {343--371},
	ISSN = {0303-6812,1432-1416},
	MRCLASS = {92C17 (35C07 35K40 35Q92 92D25)},
	MRNUMBER = {4062823},
	DOI = {10.1007/s00285-019-01391-y},
	URL = {https://doi.org/10.1007/s00285-019-01391-y},
}

@article{CGC78,
	ISSN = {00368733, 19467087},
	author = {J. William Costerton and Gill G. Geesey and Kuo-Joan Cheng},
	journal = {Scientific American},
	number = {1},
	pages = {86--95},
	publisher = {Scientific American, a division of Nature America, Inc.},
	title = {How Bacteria Stick},
	urldate = {2025-10-19},
	volume = {238},
	year = {1978}
}

@article{DH83,
	title={Selection in chemostats},
	author={Dykhuizen, Daniel E. and Hartl, Daniel L.},
	journal={Microbiological Reviews},
	volume={47},
	number={2},
	pages={150--168},
	year={1983},
	doi ={10.1128/mr.47.2.150-168.1983}
}

@article{EPL01,
	author = {Eberl, Herman J. and Parker, David F. and Van Loosdrecht, Mark C. M.},
	title = {A new deterministic spatio-temporal continuum model for biofilm development},
	journal = {Computational and Mathematical Methods in Medicine},
	volume = {3},
	number = {3},
	pages = {429794},
	doi = {https://doi.org/10.1080/10273660108833072},
	url = {https://onlinelibrary.wiley.com/doi/abs/10.1080/10273660108833072},
	year = {2001}
}

@article{FFD08,
	author = {Paul G. Falkowski  and Tom Fenchel  and Edward F. Delong },
	title = {The Microbial Engines That Drive {E}arth's Biogeochemical Cycles},
	journal = {Science},
	volume = {320},
	number = {5879},
	pages = {1034-1039},
	year = {2008},
	doi = {10.1126/science.1153213},
	URL = {https://www.science.org/doi/abs/10.1126/science.1153213},
	eprint = {https://www.science.org/doi/pdf/10.1126/science.1153213},
}

@article{FR12,
	title = {Microbial interactions: from networks to models},
	volume = {10},
	issn = {1740-1534},
	url = {https://doi.org/10.1038/nrmicro2832},
	doi = {10.1038/nrmicro2832},
	number = {8},
	journal = {Nature Reviews Microbiology},
	author = {Faust, Karoline and Raes, Jeroen},
	month = aug,
	year = {2012},
	pages = {538--550},
}

@article{FWSSRK16,
	title = {Biofilms: an emergent form of bacterial life},
	volume = {14},
	issn = {1740-1534},
	url = {https://doi.org/10.1038/nrmicro.2016.94},
	doi = {10.1038/nrmicro.2016.94},
	number = {9},
	journal = {Nature Reviews Microbiology},
	author = {Flemming, Hans-Curt and Wingender, Jost and Szewzyk, Ulrich and Steinberg, Peter and Rice, Scott A. and Kjelleberg, Staffan},
	year = {2016},
	pages = {563--575},
}

@incollection{HL14,
	address = {Cham},
	title = {Modeling of Biofilm Systems: A Review},
	isbn = {978-3-319-09695-7},
	url = {https://doi.org/10.1007/10_2014_275},
	booktitle = {Productive {Biofilms}},
	publisher = {Springer International Publishing},
	author = {Horn, Harald and Lackner, Susanne},
	editor = {Muffler, Kai and Ulber, Roland},
	year = {2014},
	doi = {10.1007/10_2014_275},
	pages = {53--76},
}

@book{HV03,
	AUTHOR = {Hundsdorfer, Willem and Verwer, Jan},
	TITLE = {Numerical solution of time-dependent
	advection-diffusion-reaction equations},
	SERIES = {Springer Series in Computational Mathematics},
	VOLUME = {33},
	PUBLISHER = {Springer-Verlag, Berlin},
	YEAR = {2003},
	PAGES = {x+471},
	ISBN = {3-540-03440-4},
	MRCLASS = {65-02 (65L05 65Mxx)},
	MRNUMBER = {2002152},
	MRREVIEWER = {Ian\ Gladwell},
	DOI = {10.1007/978-3-662-09017-6},
	URL = {https://doi.org/10.1007/978-3-662-09017-6},
}

@inproceedings{KHE09,
	title = {A Nonlinear Master Equation for a Degenerate Diffusion Model of Biofilm Growth},
	isbn = {978-3-642-01970-8},
	booktitle = {Computational {Science} – {ICCS} 2009},
	publisher = {Springer Berlin Heidelberg},
	author = {Khassehkhan, Hassan and Hillen, Thomas and Eberl, Hermann J.},
	editor = {Allen, Gabrielle and Nabrzyski, Jaroslaw and Seidel, Edward and van Albada, Geert Dick and Dongarra, Jack and Sloot, Peter M. A.},
	year = {2009},
	pages = {735--744},
}

@article{K04,
	title = {Biofilms promote altruism},
	volume = {150},
	issn = {1465-2080},
	url = {https://www.microbiologyresearch.org/content/journal/micro/10.1099/mic.0.26829-0},
	doi = {https://doi.org/10.1099/mic.0.26829-0},
	number = {8},
	journal = {Microbiology},
	author = {Kreft, Jan-Ulrich},
	year = {2004},
	pages = {2751--2760},
}

@article{KBW98,
	title = {{BacSim}, a simulator for individual-based modelling of bacterial colony growth},
	volume = {144},
	issn = {1465-2080},
	url = {https://www.microbiologyresearch.org/content/journal/micro/10.1099/00221287-144-12-3275},
	doi = {https://doi.org/10.1099/00221287-144-12-3275},
	number = {12},
	journal = {Microbiology},
	author = {Kreft, Jan-Ulrich and Booth, Ginger and Wimpenny, Julian W. T.},
	year = {1998},
	pages = {3275--3287},
}

@article{KPWL01,
	title = {Individual-based modelling of biofilms},
	volume = {147},
	issn = {1465-2080},
	url = {https://www.microbiologyresearch.org/content/journal/micro/10.1099/00221287-147-11-2897},
	doi = {https://doi.org/10.1099/00221287-147-11-2897},
	number = {11},
	journal = {Microbiology},
	author = {Kreft, Jan-Ulrich and Picioreanu, Cristian and Wimpenny, Julian W. T. and van Loosdrecht, Mark C. M.},
	year = {2001},
	pages = {2897--2912},
}

@article{MSSG23,
	title = {Competitive and substrate limited environments drive metabolic heterogeneity for comammox {Nitrospira}},
	volume = {3},
	issn = {2730-6151},
	url = {https://doi.org/10.1038/s43705-023-00288-8},
	doi = {10.1038/s43705-023-00288-8},
	number = {1},
	journal = {ISME Communications},
	author = {Martinez-Rabert, Eloi and Smith, Cindy J. and Sloan, William T. and Gonzalez-Cabaleiro, Rebeca},
	year = {2023},
	pages = {91},
}

@article{MASSG22,
	title = {Environmental and ecological controls of the spatial distribution of microbial populations in aggregates},
	volume = {18},
	url = {https://doi.org/10.1371/journal.pcbi.1010807},
	doi = {10.1371/journal.pcbi.1010807},
	number = {12},
	journal = {PLOS Computational Biology},
	author = {Martinez-Rabert, Eloi and van Amstel, Chiel and Smith, Cindy and Sloan, William T. and Gonzalez-Cabaleiro, Rebeca},
	year = {2022},
	pages = {e1010807},
}

@article{MFDPPE18,
	title = {Continuum and discrete approach in modeling biofilm development and structure: a review},
	volume = {76},
	issn = {1432-1416},
	url = {https://doi.org/10.1007/s00285-017-1165-y},
	doi = {10.1007/s00285-017-1165-y},
	number = {4},
	journal = {Journal of Mathematical Biology},
	author = {Mattei, Maria R. and Frunzo, Luigi and D’Acunto,  Berardino and Pechaud, Yoan and Pirozzi, Francesco and Esposito, Giovanni},
	year = {2018},
	pages = {945--1003},
}

@article{PF22,
	author = {Jacob D. Palmer  and Kevin R. Foster },
	title = {Bacterial species rarely work together},
	journal = {Science},
	volume = {376},
	number = {6593},
	pages = {581-582},
	year = {2022},
	doi = {10.1126/science.abn5093},
	URL = {https://www.science.org/doi/abs/10.1126/science.abn5093},
	eprint = {https://www.science.org/doi/pdf/10.1126/science.abn5093},
}

@article{PKL04,
	author = {Cristian Picioreanu and Jan-Ulrich Kreft and Mark C. M. van Loosdrecht},
	title = {Particle-Based Multidimensional Multispecies Biofilm Model},
	journal = {Applied and Environmental Microbiology},
	volume = {70},
	number = {5},
	pages = {3024-3040},
	year = {2004},
	doi = {10.1128/AEM.70.5.3024-3040.2004},
	URL = {https://journals.asm.org/doi/abs/10.1128/aem.70.5.3024-3040.2004},
	eprint = {https://journals.asm.org/doi/pdf/10.1128/aem.70.5.3024-3040.2004},
}

@article{W90,
	ISSN = {00357596, 19453795},
	URL = {http://www.jstor.org/stable/44237622},
	author = {Paul Waltman},
	journal = {The Rocky Mountain Journal of Mathematics},
	number = {4},
	pages = {777--807},
	publisher = {Rocky Mountain Mathematics Consortium},
	title={Coexistence in chemostat-like models},
	volume = {20},
	year = {1990}
}

@article{WZ10,
	title = {Review of mathematical models for biofilms},
	journal = {Solid State Communications},
	volume = {150},
	number = {21},
	pages = {1009-1022},
	year = {2010},
	issn = {0038-1098},
	doi = {https://doi.org/10.1016/j.ssc.2010.01.021},
	url = {https://www.sciencedirect.com/science/article/pii/S0038109810000281},
	author = {Qi Wang and Tianyu Zhang},
}

@article{SK16,
	title = {The nitrogen cycle},
	journal = {Current Biology},
	volume = {26},
	number = {3},
	pages = {R94-R98},
	year = {2016},
	issn = {0960-9822},
	doi = {https://doi.org/10.1016/j.cub.2015.12.021},
	url = {https://www.sciencedirect.com/science/article/pii/S0960982215015183},
	author = {Lisa Y. Stein and Martin G. Klotz},
}

@article {OS97,
	AUTHOR = {Othmer, Hans G. and Stevens, Angela},
	TITLE = {Aggregation, blowup, and collapse: the {ABC}s of taxis in
	reinforced random walks},
	JOURNAL = {SIAM Journal on Applied Mathematics},
	VOLUME = {57},
	YEAR = {1997},
	NUMBER = {4},
	PAGES = {1044--1081},
	ISSN = {0036-1399},
	MRCLASS = {92B05 (60J15)},
	MRNUMBER = {1462051},
	MRREVIEWER = {Istv\'an\ Ratk\'o},
	DOI = {10.1137/S0036139995288976},
	URL = {https://doi.org/10.1137/S0036139995288976},
}

@article{V20,
	author    = {Vinita Vishwakarma},
	title     = {Impact of environmental biofilms: Industrial components and its remediation},
	journal   = {Journal of Basic Microbiology},
	volume    = {60},
	number    = {3},
	pages     = {198--206},
	year      = {2020},
	doi       = {10.1002/jobm.201900569},
	issn      = {1521-4028},
}

@article{WW98,
	title = {Diffusion and reaction in biofilms},
	journal = {Chemical Engineering Science},
	volume = {53},
	number = {3},
	pages = {397-425},
	year = {1998},
	issn = {0009-2509},
	doi = {https://doi.org/10.1016/S0009-2509(97)00319-9},
	url = {https://www.sciencedirect.com/science/article/pii/S0009250997003199},
	author = {Brian D. Wood and Stephen Whitaker},
}

@article{WW99,
	author = {Wood, Brian D. and Whitaker, Stephen},
	title = {Cellular growth in biofilms},
	journal = {Biotechnology and Bioengineering},
	volume = {64},
	number = {6},
	pages = {656-670},
	doi = {https://doi.org/10.1002/(SICI)1097-0290(19990920)64:6<656::AID-BIT4>3.0.CO;2-N},
	url = {https://analyticalsciencejournals.onlinelibrary.wiley.com/doi/abs/10.1002/%28SICI%291097-0290%2819990920%2964%3A6%3C656%3A%3AAID-BIT4%3E3.0.CO%3B2-N},
	year = {1999}
}

@book{WEMNPRL05,
	title = {Mathematical Modeling of Biofilms},
	author = {Oskar Wanner and Herman J. Eberl and Eberhard Morgenroth and Daniel R. Noguera and Cristian Picioreanu and Bruce E. Rittmann and {van Loosdrecht}, Mark C.M.},
	year = {2005},
	isbn = {1843390876},
	volume = {IWA Scientific and Technical Reports},
	publisher = {International Water Association (IWA)},
}

@incollection{SJKH14,
	author = {Sagmeister, Patrick and Jazini, Mohammadhadi and Klein, Joachim and Herwig, Christoph},
	isbn = {9783527683321},
	title = {Bacterial Suspension Cultures},
	booktitle = {Industrial Scale Suspension Culture of Living Cells},
	chapter = {1},
	pages = {40-93},
	doi = {https://doi.org/10.1002/9783527683321.ch01},
	url = {https://onlinelibrary.wiley.com/doi/abs/10.1002/9783527683321.ch01},
	year = {2014},
}

@article{RSRW15,
	title = {A Mixed-Culture Biofilm Model with Cross-Diffusion},
	volume = {77},
	issn = {1522-9602},
	url = {https://doi.org/10.1007/s11538-015-0117-1},
	doi = {10.1007/s11538-015-0117-1},
	number = {11},
	journal = {Bulletin of Mathematical Biology},
	author = {Rahman, Kazi A. and Sudarsan, Rangarajan and Eberl, Hermann J.},
	year = {2015},
	pages = {2086--2124},
}

@article{AK07,
	title = {A Multidimensional Multispecies Continuum Model for Heterogeneous Biofilm Development},
	volume = {69},
	issn = {1522-9602},
	url = {https://doi.org/10.1007/s11538-006-9168-7},
	doi = {10.1007/s11538-006-9168-7},
	number = {2},
	journal = {Bulletin of Mathematical Biology},
	author = {Alpkvista, Erik and Klapper, Isaac},
	year = {2007},
	pages = {765--789},
}

@article{KD02,
	author = {Klapper, Isaac and Dockery, Jack},
	title = {Finger Formation in Biofilm Layers},
	journal = {SIAM Journal on Applied Mathematics},
	volume = {62},
	number = {3},
	pages = {853-869},
	year = {2002},
	doi = {10.1137/S0036139900371709},
	URL = {https://doi.org/10.1137/S0036139900371709}
}

@article{DLP15,
	title = {Complete nitrification by {Nitrospira} bacteria},
	volume = {528},
	issn = {1476-4687},
	url = {https://doi.org/10.1038/nature16461},
	doi = {10.1038/nature16461},
	number = {7583},
	journal = {Nature},
	author = {Daims, Holger and Lebedeva, Elena V. and Pjevac, Petra and Han, Ping and Herbold, Craig and Albertsen, Mads and Jehmlich, Nico and Palatinszky, Marton and Vierheilig, Julia and Bulaev, Alexandr and Kirkegaard, Rasmus H. and von Bergen, Martin and Rattei, Thomas and Bendinger, Bernd and Nielsen, Per H. and Wagner, Michael},
	year = {2015},
	pages = {504--509},
}

@article{GGMVL18,
	title = {Biofilms in the Food Industry: Health Aspects and Control Methods.},
	volume = {9},
	issn = {1664-302X},
	doi = {10.3389/fmicb.2018.00898},
	journal = {Frontiers in Microbiology},
	author = {Galié, Serena and García-Gutiérrez, Coral and Miguélez, Elisa M. and Villar, Claudio J. and Lombó, Felipe},
	year = {2018},
	pmid = {29867809},
	pmcid = {PMC5949339},
	pages = {898},
}

@article{HCCPK16,
	title = {Advancing microbial sciences by individual-based modelling},
	volume = {14},
	issn = {1740-1534},
	url = {https://doi.org/10.1038/nrmicro.2016.62},
	doi = {10.1038/nrmicro.2016.62},
	number = {7},
	journal = {Nature Reviews Microbiology},
	author = {Hellweger, Ferdi L. and Clegg, Robert J. and Clark, James R. and Plugge, Caroline M. and Kreft, Jan-Ulrich},
	year = {2016},
	pages = {461--471},
}

@article{KSANOKJL15,
	title = {Complete nitrification by a single microorganism},
	volume = {528},
	issn = {1476-4687},
	url = {https://doi.org/10.1038/nature16459},
	doi = {10.1038/nature16459},
	number = {7583},
	journal = {Nature},
	author = {van Kessel, Maartje A. H. J. and Speth, Daan R. and Albertsen, Mads and Nielsen, Per H. and Op den Camp, Huub J. M. and Kartal, Boran and Jetten, Mike S. M. and Lücker, Sebastian},
	year = {2015},
	pages = {555--559},
}

@article{LL20,
	title = {More support for {Earth}’s massive microbiome},
	volume = {15},
	issn = {1745-6150},
	url = {https://doi.org/10.1186/s13062-020-00261-8},
	doi = {10.1186/s13062-020-00261-8},
	number = {1},
	journal = {Biology Direct},
	author = {Lennon, Jay T. and Locey, Kenneth J.},
	year = {2020},
	pages = {5},
}

@article {CLM25,
	AUTHOR = {Carrillo, Jos\'{e} A. and Lorenzi, Tommaso and Macfarlane,
	Fiona R.},
	TITLE = {Spatial segregation across travelling fronts in
	individual-based and continuum models for the growth of
	heterogeneous cell populations},
	JOURNAL = {Bulletin of Mathematical Biology},
	VOLUME = {87},
	YEAR = {2025},
	NUMBER = {6},
	PAGES = {Paper No. 77, 40},
	ISSN = {0092-8240,1522-9602},
	MRCLASS = {92C37 (92D25)},
	MRNUMBER = {4907570},
	MRREVIEWER = {Gabriel\ L\'opez Garza},
	DOI = {10.1007/s11538-025-01452-y},
}

@article {Am90,
	AUTHOR = {Amann, Herbert},
	TITLE = {Dynamic theory of quasilinear parabolic equations. {II}.
	{R}eaction-diffusion systems},
	JOURNAL = {Differential Integral Equations},
	VOLUME = {3},
	YEAR = {1990},
	NUMBER = {1},
	PAGES = {13--75},
	ISSN = {0893-4983},
	MRCLASS = {35K55 (58E07 92A09)},
}

@article {HS22,
	AUTHOR = {Hissink Muller, Victor and Sonner, Stefanie},
	TITLE = {Well-posedness of singular-degenerate porous medium type
	equations and application to biofilm models},
	JOURNAL = {Journal of Mathematical Analysis and Applications},
	VOLUME = {509},
	YEAR = {2022},
	NUMBER = {1},
	PAGES = {Paper No. 125894, 34},
	ISSN = {0022-247X,1096-0813},
	DOI = {10.1016/j.jmaa.2021.125894},
	URL = {https://doi.org/10.1016/j.jmaa.2021.125894},
}

@article {SEE15,
	AUTHOR = {Sonner, Stefanie and Efendiev, Messoud A. and Eberl, Hermann
	J.},
	TITLE = {On the well-posedness of mathematical models for
	multicomponent biofilms},
	JOURNAL = {Mathematical Methods in the Applied Sciences},
	VOLUME = {38},
	YEAR = {2015},
	NUMBER = {17},
	PAGES = {3753--3775},
	ISSN = {0170-4214,1099-1476},
	MRCLASS = {92D40 (35B30 35K51)},
	MRNUMBER = {3434321},
	DOI = {10.1002/mma.3315},
	URL = {https://doi.org/10.1002/mma.3315},
}

@article{PT2008,
title = {Pattern formation in {P}seudomonas aeruginosa biofilms},
author = {Matthew R. Parsek and Tim Tolker-Nielsen},
journal = {Current Opinion in Microbiology},
volume = {11},
number = {6},
pages = {560-566},
year = {2008},
issn = {1369-5274},
doi = {https://doi.org/10.1016/j.mib.2008.09.015},
url = {https://www.sciencedirect.com/science/article/pii/S1369527408001331}
}

@article{VWKK2014,
  title={Density of founder cells affects spatial pattern formation and cooperation in {B}acillus subtilis biofilms},
  author={van Gestel, Jordi and Weissing, Franz J. and Kuipers, Oscar P. and Kov{\'a}cs, Akos T.},
  journal={The ISME Journal},
  volume={8},
  number={10},
  pages={2069--2079},
  year={2014},
  publisher={Oxford University Press}
}

@article{ZWWT2017,
  title={Emergent pattern formation in an interstitial biofilm},
  author={Zachreson, Cameron and Wolff, Christian and Whitchurch, Cynthia B. and Toth, Milos},
  journal={Physical Review E},
  volume={95},
  number={1},
  pages={012408},
  year={2017},
  publisher={APS}
}

@article{LDHD2014,
  title={Traveling waves in response to a diffusing quorum sensing signal in spatially-extended bacterial colonies},
  author={Langebrake, Jessica B. and Dilanji, Gabriel E. and Hagen, Stephen J. and De Leenheer, Patrick},
  journal={Journal of Theoretical Biology},
  volume={363},
  pages={53--61},
  year={2014},
  publisher={Elsevier}
}

@article{NS16,
	title = {Mechanically-driven spreading of bacterial populations},
	journal = {Communications in Nonlinear Science and Numerical Simulation},
	volume = {35},
	pages = {88-96},
	year = {2016},
	doi = {https://doi.org/10.1016/j.cnsns.2015.10.026},
	author = {Waipot Ngamsaad and Suthep Suantai},
}

@article{WHBI25,
  title={Individual-based modeling unravels spatial and social interactions in bacterial communities},
  author={Wang, Jian and Hashem, Ihab and Bhonsale, Satyajeet and van Impe, Jan F. M.},
  journal={The ISME Journal},
  volume={19},
  number={1},
  pages={wraf116},
  year={2025},
  publisher={Oxford University Press}
}

@article{BL15,
  title={Bacterial social interactions drive the emergence of differential spatial colony structures},
  author={Blanchard, Andrew E and Lu, Ting},
  journal={BMC Systems Biology},
  volume={9},
  number={1},
  pages={59},
  year={2015},
  publisher={Springer}
}

\end{document}